\RequirePackage[figuresright]{rotating}
\documentclass{amsart}
\usepackage{amsthm,tikz,tikz-cd,quiver,refcount,hyperref}
\pdfoutput=1

\usepackage{color}

\title{Localisable monads}
\author{Carmen Constantin}
\address{University of Edinburgh, United Kingdom}
\email{carmen.constantin@ed.ac.uk}
\author{Nuiok Dicaire}
\address{University of Edinburgh, United Kingdom}
\email{nuiok.dicaire@ed.ac.uk}
\author{Chris Heunen}
\address{University of Edinburgh, United Kingdom}
\email{chris.heunen@ed.ac.uk}
\date{\today}
\thanks{We thank Rui Soares Barbosa, Robert Furber, and Nesta van der Schaaf for useful discussions.}

\theoremstyle{plain}
\newtheorem{theorem}{Theorem}
\newtheorem{proposition}[theorem]{Proposition}
\newtheorem{corollary}[theorem]{Corollary}
\newtheorem{lemma}[theorem]{Lemma}
\theoremstyle{definition}
\newtheorem{definition}[theorem]{Definition}
\newtheorem{example}[theorem]{Example}
\newtheorem{remark}[theorem]{Remark}

\newcommand{\restate}[1]{
  \setcounterref{theorem}{#1}
  \addtocounter{theorem}{-1}
}

\DeclareMathOperator{\ZI}{ZI}
\DeclareMathOperator{\str}{st}
\DeclareMathOperator{\colim}{colim}
\DeclareMathOperator{\Alg}{Alg}
\DeclareMathOperator{\Kl}{Kl}
\newcommand{\restrict}[1]{\ensuremath{\|_{#1}}}
\newcommand{\Restrict}[1]{\ensuremath{\|^{#1}}}
\newcommand{\cat}[1]{\ensuremath{\mathbf{#1}}}
\newcommand{\op}{\ensuremath{^{\mathrm{op}}}}
\newcommand{\downset}{\ensuremath{\mathop{\downarrow}\!}}
\newcommand{\cC}{\ensuremath{\cat{C}}}
\newcommand{\formal}[1]{\ensuremath{\overline{#1}}}
\newcommand{\utv}{{\ensuremath{u\leq v}}}
\newcommand{\uov}{{\ensuremath{u\otimes v}}}
\newcommand{\id}{\ensuremath{\text{id}}}

\newcommand{\Guv}{{\ensuremath{\cat{C}\restrict{u\leq v}}}}
\newcommand{\Gu}{{\ensuremath{\cat{C}\restrict{u\leq 1}}}}

\newcommand{\Fuv}{{\ensuremath{\cat{C}\Restrict{u\leq v}}}}
\newcommand{\Fu}{{\ensuremath{\cat{C}\Restrict{u\leq 1}}}}

\newcommand{\Fuw}{{\ensuremath{\cat{C}\Restrict{u\leq w}}}}
\newcommand{\Fvw}{{\ensuremath{\cat{C}\Restrict{v\leq w}}}}
\newcommand{\Guw}{{\ensuremath{\cat{C}\restrict{u\leq w}}}}
\newcommand{\Gvw}{{\ensuremath{\cat{C}\restrict{v\leq w}}}}

\begin{document}
\maketitle

\begin{abstract}
Monads govern computational side-effects in programming semantics. They can be combined in a ``bottom-up'' way to handle several instances of such effects. Indexed monads and graded monads do this in a modular way. 
Here, instead, we equip monads with fine-grained structure in a ``top-down'' way, using techniques from tensor topology.
This provides an intrinsic theory of local computational effects without  needing to know how constituent effects interact beforehand.

Specifically, any monoidal category decomposes as a sheaf of local categories over a base space. We identify a notion of localisable monads which characterises when a monad decomposes as a sheaf of monads.
Equivalently, localisable monads are formal monads in an appropriate presheaf 2-category, whose algebras we characterise.
Three extended examples demonstrate how localisable monads can interpret the base space as locations in a computer memory, as sites in a network of interacting agents acting concurrently, and as time in stochastic processes.
\end{abstract}

\section{Introduction}

The computation of some desired value may influence parts of the environment in which the computation occurs that are separate from the value itself. Rather than being accidental byproducts, several modern programming platforms harness such \emph{computational side-effects} to structure computations in a modular way~\cite{plotkinpower:overview,plotkinpretnar:handlers}. The most well-known use is via \emph{monads}~\cite{moggi:monads,plotkinpower:monads}, which let one analyse a computational effect apart from the rest of the computation.

A computation may use more than one effect. The corresponding monads can then be combined using \emph{distributive laws} into a single monad in a ``bottom-up'' fashion~\cite{
hylandplotkinpower:combiningeffects,beck:distributivelaws,zwart:thesis}. 
This combination may involve other formalisms such as Lawvere theories~\cite{power:indexedlawvere,power:local}, but we focus on monads here.
An especially interesting case is when many instances of effects of the same kind are in play~\cite{staton:instanceseffects}. The bottom-up nature comes out in the fact that the base category on which the monad lives is highly structured; usually it is a cartesian category of presheaves.

A related use of monads is to have several layers of granularity to an effect. Indexed monads and \emph{graded monads} then model for example different levels of access to a computational effect~\cite{fujiikatsumatamellies:gradedmonads,miliuspattinsonschroeder:gradedmonads}. Again this is usually conceived of in a ``bottom-up'' fashion, where one specifies the behaviour at each level and then adds interplay between the levels.

In this article we take the opposite, ``top-down'', approach. We start with a single monad on a category with some structure, and then ask when and how that monad is the combination of constituent monads. 
This work is a first step towards an \emph{intrinsic} theory of computational effects, one that doesn't need to specify in detail how constituent effects have to interact in advance. In particular, we do not postulate that the base category consists of presheaves, which is a consequence rather than an assumption.

To do so, we follow the programme of \emph{tensor topology}, by observing that any monoidal category comes equipped with a notion of base space over which the category decomposes~\cite{enriquemolinerheunentull:tensortopology,barbosaheunen:sheaves,enriquemolinerheunentull:space,heunenlemay:tensorrestriction}. This ``spatial'' aspect can be cleanly separated: any monoidal category embeds into a category of global sections of a sheaf of so-called local monoidal categories (see Theorems~\ref{thm:embedding} and~\ref{thm:sheafrepresentation} below). This is recalled in Section~\ref{sec:tensortopology}.

Our main question is when and how a monad on a monoidal category respects this decomposition in the sense that it corresponds to a sheaf of monads on the local categories. The answer is a \emph{localisable monad}, discussed in Section~\ref{sec:localisablemonads}. To connect back to the ``bottom-up'' approach, we then characterise such monads as \emph{formal monads}~\cite{street:formalmonads} in a (pre)sheaf category in Section~\ref{sec:formalmonads}. This opens a way to analyse the (Kleisli) algebras for localisable monads, which we do in Section~\ref{sec:algebras}.

The breadth of this approach is demonstrated in Section~\ref{sec:examples}, where we work out three extended examples. They show a range of how localisable monads may interpret the base space:
as locations in a computer memory governed by a \emph{local state} monad;
as sites in a network of interacting agents governed by a monad inspired by the \emph{pi calculus}; 
and as moments in time governed by a monad of \emph{stochastic processes}.

Section~\ref{sec:conclusion} concludes, and Appendix~\ref{sec:proofs} gives proofs that were deferred from the main text.

\section{Tensor topology}\label{sec:tensortopology}

This section summarises necessary notions from tensor topology. We have to be brief, and for more details we refer the reader to~\cite{enriquemolinerheunentull:tensortopology,barbosaheunen:sheaves,heunenlemay:tensorrestriction,enriquemolinerheunentull:space}.
To save space we will not use the graphical calculus for monoidal categories~\cite{heunenvicary:cqm}, but will not be careful in denoting coherence isomorphisms in this section.
The following notions and results hold for arbitrary monoidal categories, but for simplicity we deal here with the symmetric monoidal case only.

\begin{definition}
  A \emph{central idempotent}	in a symmetric monoidal category is a morphism $u \colon U \to I$ such that $\rho_U \circ (U \otimes u) = \lambda_U \circ (u \otimes U) \colon U \otimes U \to U$ and this map is invertible. We identify two central idempotents $u \colon U \to I$ and $v \colon V \to I$ when there is an isomorphism $m \colon U \to V$ satisfying $u = v \circ m$. Write $\ZI(\cat{C})$ for the collection of central idempotents of $\cat{C}$.
\end{definition}

A central idempotent $u \colon U \to I$ is completely determined by its domain $U$. The central idempotents always form a (meet-)semilattice. The order is defined by $u \leq v$ if and only if $u=v\circ m$ for some morphism $m \colon U \to V$. The meet is given $u \wedge v = \lambda_I \circ (u \otimes v) \colon U \otimes V \to I$. The largest central idempotent is the identity $1 \colon I \to I$.

\begin{example}\label{ex:semilattice}
  Consider a (meet-)semilattice $(L,\wedge,1)$ as a symmetric monoidal category $\cat{C}$: objects of $\cat{C}$ are elements of $L$, there is a morphism $u \to v$ if and only if $u \leq v$, and $u \otimes v = u \wedge v$. Then $\ZI(\cat{C}) \simeq L$. In fact, $\ZI$ is a functor that is right adjoint to the inclusion of the category of semilattices into the category of symmetric monoidal categories.
\end{example}

\begin{example}\label{ex:cartesian}
  If $\cat{C}$ is cartesian -- that is, tensor products are in fact categorical products -- then central idempotents are exactly subterminal objects: objects $U$ whose unique morphism $! \colon U \to 1$ to the terminal object is monic.

  In particular, if $X$ is any topological space, the category of sheaves over $X$ has as central idempotent semilattice the collection of open sets $U \subseteq X$ under intersection.
\end{example}

\begin{example}\label{ex:hilbertmodules}
  If $X$ is a locally compact Hausdorff topological space, the category of Hilbert modules over $C_0(X)$ is symmetric monoidal. It is equivalent to the category of fields of Hilbert spaces over $X$, and its central idempotents correspond to open subsets $U \subseteq X$.
\end{example}

Because of the previous examples, we can think of central idempotents as open subsets of a hidden base space that any symmetric monoidal category comes equipped with. Tensor topology develops general accompanying notions of locality, restriction, and support. For example, we can restrict attention to the `part of the category that lives over an open set', as follows.

\begin{proposition}\label{prop:Cu}
  For every central idempotent $u$ in a symmetric monoidal category $\cat{C}$, there is a symmetric monoidal category $\cat{C}\restrict{u}$ where:
  \begin{itemize}
	  \item objects are as in $\cat{C}$;
	  \item morphisms $A \to B$ are morphisms $A \otimes U \to B$ in $\cat{C}$;
	  \item composition of $f \colon A \otimes U \to B$ and $g \colon B \otimes U \to C$ is $g \circ (f \otimes U) \circ (A \otimes U \otimes u)^{-1} \colon A \otimes U \to C$;
	  \item the identity on $A$ is given by $A \otimes u$;
	  \item tensor product of objects is as in $\cat{C}$;
	  \item tensor product of morphisms $f \colon A \otimes U \to B$ and $f' \colon A' \otimes U \to B'$ is $(f \otimes f') \circ (A \otimes \sigma_{A',U} \otimes U) \circ (A \otimes A' \otimes U \otimes u)^{-1} \colon A \otimes A' \otimes U \to B \otimes B'$.\qed
  \end{itemize}
\end{proposition}

\begin{remark}\label{rem:Cuequivalent}
	In $\cat{C}\restrict{u}$, any object $A$ is isomorphic to $A \otimes U$: the isomorphism and its inverse are given by the identity $A \otimes U \to A \otimes U$ in $\cat{C}$ and $A \otimes u \otimes u \colon A \otimes U \otimes U \to A$. 
\end{remark}

\begin{example}
  In the category $\cat{C}$ of sheaves over a topological space $X$, central idempotents $u$ correspond to open subsets $U \subseteq X$ as in Example~\ref{ex:cartesian}. The category $\cat{C}\restrict{u}$ is then equivalent to the category of sheaves over $U$.
\end{example}

The intuition of a category $\cat{C}$ `living over' open subsets is further strengthened by the following lemma, that says we can pass between the part of a category living over a larger open subset and the part living over a smaller open subset.

\begin{lemma}\label{lem:adjunctionCu}
  If $u \leq v$ are central idempotents in $\cat{C}$, with $u = v \circ m$, there is an adjunction:
  \[\begin{tikzpicture}
	  \node (l) at (0,0) {$\cat{C}\restrict{u}$};
	  \node (r) at (4,0) {$\cat{C}\restrict{v}$};
	  \draw[draw=none] (l) to node{$\perp$} (r);
	  \draw[->] (l) to[out=15,in=165] node[above]{$\cat{C}\Restrict{u \leq v}$} (r);
	  \draw[->] (r) to[out=-165,in=-15] node[below]{$\cat{C}\restrict{u \leq v}$} (l);
  \end{tikzpicture}\]
  The functor $\cat{C}\restrict{u \leq v}$ is given by $A \mapsto A$ and $f \mapsto f \circ (A \otimes m)$ and is strict monoidal.
  The functor $\cat{C}\Restrict{u \leq v}$ is given by $A \mapsto A \otimes U$ and $f \mapsto (f \otimes U) \circ (A \otimes u \otimes U)^{-1} \circ (A \otimes U \otimes v)$ and is oplax monoidal. The unit of the adjunction is an isomorphism.
\end{lemma}
\begin{proof}
  See~\cite[Lemmas~5.4 and~5.5]{barbosaheunen:sheaves}.
\end{proof}

To make the intuition built up so far completely rigorous, we now summarise a series of results saying that any symmetric monoidal category may be regarded as a sheaf of monoidal categories over a base topological space. To state them, we need to introduce mild conditions on the central idempotents being respected by tensor products.

\begin{definition}
  A symmetric monoidal category $\cat{C}$ is called \emph{stiff} when 
  the diagram on the left below
  is a pullback for any object $A$ and central idempotents $u$ and $v$.
  \[
  	\begin{tikzpicture}[xscale=3.5,yscale=1.3]
      \node (tl) at (0,1) {$A \otimes U \otimes V$};
      \node (tr) at (1,1) {$A \otimes V$};
      \node (bl) at (0,0) {$A \otimes U$};
      \node (br) at (1,0) {$A$};
      \draw[->] (tl) to (tr);
      \draw[->] (tl) to (bl);
      \draw[->] (tr) to node[right]{$A \otimes v$} (br);
      \draw[->] (bl) to node[above]{$A \otimes u$} (br);
      \draw (.1,.7) to (.15,.7) to (.15,.825);
    \end{tikzpicture}
    \qquad
		\begin{tikzpicture}[xscale=3.5,yscale=1.3]
	    \node (tl) at (0,1) {$A \otimes U \otimes V$};
	    \node (tr) at (1,1) {$A \otimes V$};
	    \node (bl) at (0,0) {$A \otimes U$};
	    \node (br) at (1,0) {$A \otimes (U \vee V)$};
	    \draw[>->] (tl) to node[above]{} (tr);
	    \draw[>->] (tl) to node[left]{} (bl);
	    \draw[>->] (tr) to node[right]{} (br);
	    \draw[>->] (bl) to node[below]{} (br);
	    \draw (.1,.7) to (.15,.7) to (.15,.825);
	    \draw (.95,.3) to (.90,.3) to (.90,.175);
	  \end{tikzpicture}
	\]
  We say $\cat{C}$ has \emph{finite universal joins} of central idempotents when it has an initial object $0$ satisfying $A \otimes 0 \simeq 0$ for all objects $A$, and $\ZI(\cat{C})$ has binary joins such that the diagram on the right above is a pullback and a pushout for all objects $A$ and central idempotents $u$ and $v$.
\end{definition}

The following theorem says that any stiff monoidal category can be freely completed with universal finite joins of central idempotents~\cite[Theorem~12.8]{barbosaheunen:sheaves}.

\begin{theorem}\label{thm:embedding}
  Any stiff symmetric monoidal category allows a strict monoidal full embedding into a symmetric monoidal category with	finite universal joins of central idempotents.
\end{theorem}

Finally, the following theorem~\cite[Theorem~8.6]{barbosaheunen:sheaves} says that any symmetric monoidal category $\cat{C}$ with universal finite joins has a particularly nice form.
It considers the semilattice of central idempotents $\ZI(\cat{C})$ as the basic opens of a topological space $X$ by taking its Zariski spectrum~\cite[Section~4]{barbosaheunen:sheaves}.

\begin{theorem}\label{thm:sheafrepresentation}
  Any symmetric monoidal category $\cat{C}$ with universal finite joins of central idempotents is monoidally equivalent to a category of global sections of a sheaf $u \mapsto \cat{C}\restrict{u}$ of local monoidal categories over $\ZI(\cat{C})$.
\end{theorem}

Here, a monoidal category $\cat{C}$ is called local when $u \vee v=1$ implies $u=1$ or $v=1$ in $\ZI(\cat{C})$. When $\ZI(\cat{C})$ is the opens of a topological space, that means there is a single focal point that all nets in the topological space converge to -- intuitively, $\cat{C}$ is local when it has no nontrivial central idempotents. Being a sheaf of local monoidal categories means that the stalks $\cat{C}\restrict{x} = \colim_{x \in u} \cat{C}\restrict{u}$ over points $x \in X$ are local monoidal categories.

It follows that any stiff symmetric monoidal category embeds into such a category of global sections. This makes precise the intuition that a symmetric monoidal category continuously varies over its base space of central idempotents.

\section{Localisable monads}\label{sec:localisablemonads}

The previous section showed how any symmetric monoidal category $\cat{C}$ may be regarded as a sheaf $\cat{C}\restrict{u}$ of local ones. In this section, we work out when a monad on $\cat{C}$ corresponds to a sheaf of monads on $\cat{C}\restrict{u}$. The crucial definition is as follows.

\begin{definition}\label{def:localisablemonad}
  A monad $T$ on a monoidal category $\cat{C}$ is called \emph{localisable} when there are morphisms $\str_{A,U} \colon T(A) \otimes U \to T(A \otimes U)$ for each object $A$ and central idempotent $u \colon U \to I$ satisfying:
  \begin{align}
    T(\rho_A) \circ \str_{A,I} &= \rho_{T(A)} \label{eq:st:unitor} \\
    T(\alpha_{A,U,V}) \circ \str_{A,U \otimes V} & = \str_{A \otimes U,V} \circ (\str_{A,U} \otimes V) \circ \alpha_{TA,U,V} \label{eq:st:associator} \\
    \eta_{A \otimes U} &= \str_{A,U} \circ (\eta_A \otimes U) \label{eq:st:unit} \\
    \mu_{A \otimes U} \circ T(\str_{A,U}) \circ \str_{T(A),U} &= \str_{A,U} \circ (\mu_A \otimes U) \label{eq:st:multiplication} \\
    \str_{A,V} \circ (T(A) \otimes m) & = T(A \otimes m) \circ \str_{A,U}
    \label{eq:st:naturalinu} \\
    \str_{B,U} \circ \big( T(f) \otimes U \big) &= T(f \otimes U) \circ \str_{A,U} \label{eq:st:naturalina} 
  \end{align}
  for any morphism $f \colon A \to B$ and central idempotents $u \colon U \to I$ and $v \colon V \to I$, and where $m \colon U \to V$ in~\eqref{eq:st:naturalinu} satisfies $u = v \circ m$.
\end{definition}

\begin{example}\label{ex:closureoperator}
  Consider a semilattice $(L,\wedge,1)$	as a symmetric monoidal category $\cat{C}$ as in Example~\ref{ex:semilattice}. A monad on $\cat{C}$ then is exactly a closure operator on $L$, that is, a function $\overline{(-)} \colon L \to L$ satisfying $u \leq \overline{u} = \overline{\overline{u}}$ and $u \leq v \implies \overline{u} \leq \overline{v}$.
  This monad is localisable if and only if $\overline{u} \wedge v \leq \overline{u \wedge v}$ for all $u,v \in L$.
  This is for example the case when $L$ is the powerset of a set $X$, and $\overline{U}$ is the closure of $U \subseteq X$ in a fixed topology on $X$.
\end{example}

\begin{example}\label{ex:strong}
  Strong monads~\cite{kock:monoidalmonads,jacobs:weakeningcontraction} are localisable: axioms~\eqref{eq:st:unitor}--\eqref{eq:st:multiplication} are a special case of the axioms for a strong monad; and axioms~\eqref{eq:st:naturalinu}--\eqref{eq:st:naturalina} follow from naturality of strength.
  Hence a monad $T$ on a symmetric monoidal closed category is localisable if $T(U \multimap A) \simeq T(U) \multimap T(A)$, namely with $\str_{A,U}$ as follows (where $\mathrm{coev}$ denotes the curry of the identity on $A\otimes U$)
  \[\begin{tikzcd}[column sep=15mm, row sep=2mm]
    {T(A) \otimes U} & {T\big(U \multimap (A \otimes U)\big) \otimes T(U)} \\
    {T(A \otimes U)} & {\big(T(U) \multimap T(A \otimes U)\big) \otimes T(U)}
    \arrow["{T(\mathrm{coev}) \otimes \eta}", from=1-1, to=1-2]
    \arrow[Rightarrow, no head, from=1-2, to=2-2]
    \arrow["{\mathrm{ev}}", from=2-2, to=2-1]
  \end{tikzcd}\]
\end{example}

\begin{example}
  It follows from Example~\ref{ex:strong} and~\cite{kock:cartesianclosedmonads} that a monad $T$ on a cartesian closed category is localisable as soon as $T(A \times B) \simeq T(A) \times T(B)$. In particular, this applies for any monad on the category of sheaves over a topological space $X$ as in Example~\ref{ex:cartesian}.
\end{example}

We will work out more examples in Section~\ref{sec:examples} below.
Next we consider the main consequence of a monad on $\cat{C}$ being localisable: it restricts to the categories $\cat{C}\restrict{u}$.

\begin{proposition}\label{prop:smallmonad}
  If $T$ is a localisable monad on $\cat{C}$ and $u$ a central idempotent, the following defines a monad $T\restrict{u}$ on $\cat{C}\restrict{u}$:
  \begin{align*}
    T\restrict{u}(A) &= T(A) & & (\eta\restrict{u})_A = \eta_A \otimes u \\
    T\restrict{u}\big(f \colon A \otimes U \to B\big) &= T(f) \circ \str_{A,U} & & (\mu\restrict{u})_A = \mu_A \otimes u    
  \end{align*}
\end{proposition}
\begin{proof}
  This is mainly a matter of unwinding definitions and being careful in which category compositions are taken. For example, the unit law $(\mu\restrict{u})_A \circ (\eta\restrict{u})_{T\restrict{u}(A)} = T(A)$ in $\cat{C}\restrict{u}$ comes down to the following diagram commuting in $\cat{C}$:
  \[\begin{tikzcd}[column sep=18mm,row sep=7mm]
  	{T(A) \otimes U \otimes U} & {T^2(A) \otimes U} \\
  	{T(A) \otimes U} & {T(A)}
  	\arrow["{\eta_{TA} \otimes U}"{description}, from=2-1, to=1-2]
  	\arrow["{\mu_A \otimes u}", from=1-2, to=2-2]
  	\arrow["{\eta_{TA} \otimes u \otimes U}", from=1-1, to=1-2]
  	\arrow["{T(A) \otimes u}"', from=2-1, to=2-2]
  	\arrow["{T(A) \otimes (u \otimes U)^{-1}}", from=2-1, to=1-1]
  \end{tikzcd}\]
  Similarly, naturality of $\eta\restrict{u}$, which is $T\restrict{u}(f) \circ (\eta\restrict{u})_A = (\eta\restrict{u})_B \circ f$ in $\cat{C}\restrict{u}$, comes down to commutativity of the following diagram in $\cat{C}$:
  \[\begin{tikzcd}[column sep = 18 mm]
  	{A \otimes U \otimes U} & {T(A) \otimes U \otimes I} & {T(A) \otimes U} \\
  	{A \otimes U \otimes U} & {T(A \otimes U) \otimes I} & {T(A \otimes U)} \\
  	{B \otimes U} & {T(B) \otimes I} & {T(B)}
  	\arrow["{\eta_A \otimes U \otimes u}", from=1-1, to=1-2]
  	\arrow["{\rho_{T(A) \otimes U}}", from=1-2, to=1-3]
  	\arrow["{\str_{A,U}}", from=1-3, to=2-3]
  	\arrow["{T(f)}", from=2-3, to=3-3]
  	\arrow["{\eta_B \otimes u}"', from=3-1, to=3-2]
  	\arrow["{\rho_{T(B)}}"', from=3-2, to=3-3]
  	\arrow["{\str_{A,U} \otimes I}"', from=1-2, to=2-2]
  	\arrow["{\rho_{T(A \otimes U})}"', from=2-2, to=2-3]
  	\arrow[Rightarrow, no head, from=1-1, to=2-1]
  	\arrow["{f \otimes U}"', from=2-1, to=3-1]
  	\arrow[from=2-1, to=2-2]
  	\arrow["{T(f) \otimes I}"', from=2-2, to=3-2]
  \end{tikzcd}\]
  Here the upper left square follows from~\eqref{eq:st:unit}, the right squares are naturality of unitors, and the lower left square is naturality of $\eta$ in $\cat{C}$.
  The other laws are verified similarly.
\end{proof}

\begin{example}
  Consider a closure operator $T(u)=\overline{u}$ on a semilattice $\cat{C}=L$ as in Example~\ref{ex:closureoperator}. Then $T_u(a)$ is simply $\overline{a}$. This is a well-defined closure operator on the pre-order $\cat{C}\restrict{u}$: if $a \wedge u \leq b$, then $\overline{a} \wedge u \leq \overline{a \wedge u} \leq \overline{b}$ because $T$ is localisable.
  Collapsing the pre-order $\cat{C}\restrict{u}$ to a partially ordered semilattice as in Remark~\ref{rem:Cuequivalent} simply gives the downset $\downset u = \{ a \in L \mid a \leq u\}$ of $u$ in $L$, and $T_u$ just becomes the restriction of the closure operator to $\downset u$.
\end{example}

Recall that a \emph{(lax) monad morphism}~\cite{street:formalmonads} from a monad $(S,\eta^S,\mu^S)$ on $\cat{C}$ to a monad $(T,\eta^T,\mu^T)$ on $\cat{D}$ consists of a functor $F \colon \cat{C} \to \cat{D}$ and a natural transformation $\varphi \colon T \circ F \Rightarrow F \circ S$ making the following two diagrams commute:
\begin{equation}\label{eq:monadmorphism}\begin{aligned}
\begin{tikzcd}[column sep=4mm, row sep=4mm]
F && {T \circ F} \\
\\
& {F \circ S}
\arrow["{\eta_F^T}", Rightarrow, from=1-1, to=1-3]
\arrow["\varphi", Rightarrow, from=1-3, to=3-2]
\arrow["{F\eta^S}"', Rightarrow, from=1-1, to=3-2]
\end{tikzcd}
\qquad\qquad
\begin{tikzcd}[column sep=7mm, row sep=8mm]
  {T^2 \circ F} & {T \circ F \circ S} & {F \circ S^2} \\
  {T \circ F} && {F \circ S}
  \arrow["{\mu^T_F}"', Rightarrow, from=1-1, to=2-1]
  \arrow["\varphi"', Rightarrow, from=2-1, to=2-3]
  \arrow["{F\mu^S}", Rightarrow, from=1-3, to=2-3]
  \arrow["{\varphi S}", Rightarrow, from=1-2, to=1-3]
  \arrow["{T\varphi}", Rightarrow, from=1-1, to=1-2]
\end{tikzcd}
\end{aligned}\end{equation}
Monads on $\cat{C}$ and their (lax) morphisms form a category $\cat{Monad}(\cat{C})$. An \emph{oplax monad morphism} has $\psi \colon F \circ S \Rightarrow T \circ F$ that respects units and multiplication instead of $\varphi$.

\begin{lemma}\label{lem:monadmorphism}
  Let $T$ be a localisable monad on $\cat{C}$.
  If $u \leq v$ are central idempotents, then the functor $\cat{C}\restrict{u \leq v}$ from Lemma~\ref{lem:adjunctionCu} is a (lax) monad morphism $T\restrict{v} \to T\restrict{u}$ with $\varphi_A = T(A) \otimes u$.
\end{lemma}
\begin{proof}
	Here we need to show the naturality of $\varphi$ and the commutativity of the diagrams~\eqref{eq:monadmorphism}. These directly follow from \eqref{eq:st:naturalinu}, bifunctoriality of the tensor product and a few commuting diagrams that can be found in Appendix~\ref{sec:proofs}.
\end{proof}

If $F \colon \cat{C} \to \cat{D}$ with $\varphi \colon T \circ F \Rightarrow F \circ S$ is a (lax) monad morphism between localisable monads $S$ and $T$, and $F$ is a (lax) monoidal functor with $\theta_{A,B} \colon F(A) \otimes F(B) \to F(A \otimes B)$, we say $(F,\varphi,\theta)$ is a (lax) \emph{morphism of localisable monads} when the following diagram commutes:
\[\begin{tikzcd}
	{TF(A) \otimes F(U)} && {T\big(F(A) \otimes F(U)\big)} && {TF(A \otimes U)} \\
	{FS(A) \otimes F(U)} && {F\big(S(A) \otimes U\big)} && {FS(A \otimes U)}
	\arrow["{\str_{FA,FU}}", from=1-1, to=1-3]
	\arrow["{T(\theta_{A,U})}", from=1-3, to=1-5]
	\arrow["{\varphi_{A,U}}", from=1-5, to=2-5]
	\arrow["{\theta_{S(A),U}}"', from=2-1, to=2-3]
	\arrow["{\varphi_A \otimes F(U)}"', from=1-1, to=2-1]
	\arrow["{\str_{A,U}}"', from=2-3, to=2-5]
\end{tikzcd}\]
In this sense, the monad morphism $T\restrict{v} \to T\restrict{u}$ of Lemma~\ref{lem:monadmorphism} is localisable.

\begin{corollary}
  If $T$ is a localisable monad on $\cat{C}$, and $u \leq v$ are central idempotents, then the functor $\cat{C}\Restrict{u \leq v}$ from Lemma~\ref{lem:adjunctionCu} is an oplax monad morphism $T\restrict{u} \to T\restrict{v}$ with $\psi_A = \str_{A,U}$.
\end{corollary}
\begin{proof}
  Applying~\cite[Theorem~9]{street:formalmonads} to Lemmas~\ref{lem:adjunctionCu} and~\ref{lem:monadmorphism}, we can compute $\psi$ as follows.
  By the adjunction, $\varphi_A \colon T\restrict{u}(\cat{C}\restrict{u \leq v}(\cat{C}\Restrict{u \leq v}(A))) \to \cat{C}\restrict{u \leq v}(T\restrict{v}(\cat{C}\Restrict{u \leq v}(A)))$ corresponds to a morphism 
  \[
    \cat{C}\Restrict{u \leq v}(T\restrict{u}(\cat{C}\restrict{u \leq v}(\cat{C}\Restrict{u \leq v}(A)))) 
    \to
    T\restrict{v}(\cat{C}\Restrict{u \leq v}(A))
  \]
  and $\psi_A \colon \cat{C}\Restrict{u \leq v}(T\restrict{u}(A)) \to T\restrict{v}(\cat{C}\Restrict{u \leq v}(A))$ is obtained by precomposing this morphism with the unit $A \to \cat{C}\restrict{u \leq v}(\cat{C}\Restrict{u \leq v}(A))$ of the adjunction. 
  Starting with $\varphi_A = T(A) \otimes u$, this gives exactly $\psi_A = \str_{A,U}$.
\end{proof}

\begin{remark}\label{rem:stalks}
  If $T$ is a localisable monad on a stiff symmetric monoidal category $\cat{C}$, and $x$ is a point of $\ZI(\cat{C})$ regarded as a topological space, we can go further and define a monad $T\restrict{x}$ on the stalk $\cat{C}\restrict{x}$. 
  The stalk $\cat{C}\restrict{x}$ is defined as the colimit of the diagram $\cat{C}\restrict{u \leq v} \colon \cat{C}\restrict{v} \to \cat{C}\restrict{u}$ ranging over all central idempotents $u \leq v$ containing the point $x$, taken in the category of symmetric monoidal categories.
  Accordingly, $T\restrict{x}$ is the colimit over the same diagram, but now taken in the category of localisable monads.
  Using the concrete description in~\cite[Definition~7.1]{barbosaheunen:sheaves} of these stalks, we can compute:
  \begin{align*}
    T\restrict{x}(A) &= T(A) & & (\eta\restrict{x})_A = [1,\eta_A \circ \rho_A] \\
    T\restrict{x}\big([u,f \colon A \otimes U \to B]\big) &= [u,T(f) \circ \str_{A,U}] & & (\mu\restrict{x})_A = [1,\mu_A \circ \rho_{T^2(A)}]
  \end{align*}
  If $x \in u$ there is a localisable monad morphism $T\restrict{u} \to T\restrict{x}$ formed by the functor $\cat{C}\restrict{x \in u} \colon \cat{C}\restrict{u} \to \cat{C}\restrict{x}$ given by $\cat{C}\restrict{x \in u}(A)=A$ and $\cat{C}\restrict{x \in u}(f \colon A \to B) = [u,f]$ with the identity natural transformation $\varphi \colon T\restrict{u} \circ \cat{C}\restrict{x \in u} \Rightarrow \cat{C}\restrict{x \in u} \circ T\restrict{x}$.

  The representation of Theorems~\ref{thm:embedding} and~\ref{thm:sheafrepresentation} is in fact functorial~\cite[Section~11]{barbosaheunen:sheaves}: a (lax) monoidal functor $T \colon \cat{C} \to \cat{C}$ corresponds to a family of stalk functors $T\restrict{x} \colon \cat{C}\restrict{x} \to \cat{C}\restrict{x}$ that are continuous in a certain sense. However, this notion of continuity is quite involved, and we will not pursue it further here.
\end{remark}

\section{Formal monads, graded monads, and indexed monads}\label{sec:formalmonads}

This section characterises localisable monads as formal monads in a certain presheaf category, and connects to graded monads and indexed monads.

\subsection{Formal monads}

We will characterise localisable monads as formal monads in the 2-category $[\ZI(\cat{C})\op,\cat{Cat}]$ with functors $\ZI(\cat{C})\op \to \cat{Cat}$ as $0$-cells, \emph{natural} transformations as $1$-cells, and modifications as $2$-cells~\cite{street:formalmonads,leinster:higher}. 
More precisely, we will define a formal monad \emph{on} the sheaf $\formal{\cat{C}} \colon \ZI(\cat{C})\op\to\cat{Cat}$ that maps a central idempotent $u$ to the category $\cC\restrict{u}$ and morphisms $u\leq v$ to the functors $\Guv\colon \cC\restrict{v} \to \cC\restrict{u}$ of Lemma~\ref{lem:adjunctionCu}.
A formal monad then consists of a natural transformation $\formal{T}\colon\formal{\cC}\Rightarrow \formal{\cC}$ and two modifications $\mu\colon\formal{T}\formal{T} \Rrightarrow \formal{T}$ and $\eta\colon\id_{\formal{\cC}}\Rrightarrow \formal{T}$ satisfying the usual monad laws. 
More precisely, the data of this formal monad consists of:
\begin{itemize}
  \item monads $(T\restrict{u},\mu\restrict{u},\eta\restrict{u})$ on $\cat{C}\restrict{u}$ for every central idempotent $u$ in $\cat{C}$;
  \item functors $\Guv \colon \cat{C}\restrict{v} \to \cat{C}\restrict{u}$ for central idempotents $u \leq v$ in $\cat{C}$;
\end{itemize}
such that the following equations hold in $\cat{C}\restrict{u}$:
\begin{equation}\label{eq:modification}
\begin{aligned}
\begin{tikzcd}[row sep=5mm, column sep = 8.3 mm]
	{\Guv(A)} && {\Guv(T\restrict{v}(A))} \\
	&& {T\restrict{u}(\Guv(A))}
	\arrow[Rightarrow, no head, from=1-3, to=2-3]
	\arrow["{\Guv((\eta\restrict{v})_A)}", from=1-1, to=1-3]
	\arrow["{(\eta\restrict{v})_{\Guv (A)}}"', from=1-1, to=2-3]
\end{tikzcd}
\quad
\begin{tikzcd}[row sep=5mm, column sep = 8.3 mm]
	{T_u^2(\Guv (A))} && {T_u(\Guv(A))} \\
	{\Guv(T_v^2(A))} && {\Guv(T_v(A))}
	\arrow["{\Guv((\mu\restrict{v})_A)}"', from=2-1, to=2-3]
	\arrow["{(\mu\restrict{u})_{\Guv (A)}}", from=1-1, to=1-3]
	\arrow[Rightarrow, no head, from=1-3, to=2-3]
	\arrow[Rightarrow, no head, from=1-1, to=2-1]
\end{tikzcd}
\end{aligned}
\end{equation}
Moreover $\formal{T}$ is natural, meaning that if $u=v\circ m$ then for any $f \colon A\to B$ in $\cC\restrict{v}$:
\begin{align}
T\restrict{u}(\Guv A) & = \Guv T\restrict{v}(A) \label{eq:equal_ob}\\
T\restrict{u}\big(\Guv f\big) & = \Guv T\restrict{v}(f).\label{eq:nat_on_mor}
\end{align}
Given the definition of $\Guv$, the first equation simply reads $T\restrict{u}(A) = T\restrict{v}(A)$. 
The following two lemmas follow from the definition of the adjoint functors $\Fuv\dashv \Guv$.

\begin{lemma}
	There is a comonad $-\otimes U$ on $\cC$ for any central idempotent $u$ of $\cC$. More generally, there is a comonad $-\otimes U$ on $\cC\restrict{v}$ for any central idempotents $u\leq v$ of $\cC$.
\end{lemma}
\begin{lemma}\label{thm:coKleisli}
	The category $\cC\restrict{u}$ is the co-Kleisli category of the comonad $-\otimes U$ on $\cC\restrict{v}$. 
\end{lemma}
It follows from Lemma \ref{thm:coKleisli} that there is a canonical adjunction between the co-Kleisli category $\cC\restrict{u}$ and category $\cC\restrict{v}$ (or the base category $\cC$ for $v=1$) given by adjoint functors $\Fuv \dashv \Guv$ such that $- \otimes U=\Fuv\circ \Guv$. These correspond to the adjoint functors defined in Lemma \ref{lem:adjunctionCu}. Further than Lemma~\ref{lem:adjunctionCu}, observe the following decomposition.

\begin{lemma}\label{lem:decompFG}
  If $u\leq v\leq w$ are central idempotents in $\cC$, the functors of Lemma~\ref{lem:adjunctionCu} satisfy:
  \[
    \begin{aligned}
      \cC\Restrict{u \leq w} &= \cC\Restrict{v \leq w} \circ \cC\Restrict{u \leq v} \\[2mm]
      \cC\restrict{u \leq w} &= \cC\restrict{u \leq v} \circ \cC\restrict{v \leq w} \\[2mm]
      \cC\Restrict{u \leq v} &= \cC\restrict{v \leq w} \circ \cC\Restrict{u \leq w}
    \end{aligned}
    \qquad\qquad
    \begin{aligned}\begin{tikzcd}[row sep = 5mm]
    	{\cC\restrict{v}} \\
    	&&&& {\cC\restrict{w}} \\
    	{\cC\restrict{u}}
    	\arrow["\Fuv"{description, pos=0.4}, curve={height=-12pt}, from=3-1, to=1-1]
    	\arrow["\Guv"{description, pos=0.4}, curve={height=-12pt}, from=1-1, to=3-1]
    	\arrow["\Fvw"{description}, curve={height=-18pt}, from=1-1, to=2-5]
    	\arrow["\Gvw"{description}, from=2-5, to=1-1]
    	\arrow["\Fuw"{description}, curve={height=18pt}, from=3-1, to=2-5]
    	\arrow["\Guw"{description}, from=2-5, to=3-1]
    \end{tikzcd}\end{aligned}
  \]
\end{lemma} 
\begin{proof}
  This follows directly from the definition of the functors.
\end{proof}

\begin{proposition}\label{thm:formal_to_loc}
	Let $\cC$ be a stiff category. Let $(\formal{T},\formal{\mu},\formal{\eta})$ be a formal monad in $[\ZI(\cat{C})^{\op}, \cat{Cat}]$ above $\formal{\cC}$ and let $u \leq v$ be central idempotents. Then the monad $T\restrict{v}$ is a localisable monad with the strength $\str_{A,U} \colon T\restrict{v}(A)\otimes U \to T\restrict{v}(A\otimes U)$ defined as the following composition in $\cC\restrict{v}$ for any object $A$ in $\cC\restrict{v}$:
\begin{equation}\label{eq:def_strength}
\begin{tikzcd}
	{T\restrict{v}(A)\otimes U=\Fuv \Guv T\restrict{v}A = \Fuv T\restrict{u}\Guv A} \\
	{\Fuv T\restrict{u} \Guv \Fuv  \Guv A = \Fuv  \Guv T\restrict{v} \Fuv  \Guv A} \\
	{T\restrict{v} \Fuv  \Guv A = T\restrict{v}(A\otimes U)}
	\arrow["{\varepsilon^{u\leq v}_{T\restrict{v} \Fuv \Guv  A}}", from=2-1, to=3-1]
	\arrow["{\Fuv  T\restrict{u}\eta^{u\leq v}_{\Guv  A}}", from=1-1, to=2-1]
\end{tikzcd}
\end{equation}
where $\eta^{u\leq v}$ and $\varepsilon^{u\leq v}$ are the unit and counit of adjunction $\Fuv \dashv \Guv$.
\end{proposition}
\begin{proof}
	We need to prove each of the axioms of Definition~\ref{def:localisablemonad}. This consist of many commutating diagrams, found in Appendix~\ref{sec:proofs}. For simplicity, the proof is laid out for the case $v=1$, but the same arguments hold for any $T\restrict{v}$ by using the relevant strength. 
\end{proof}

\begin{proposition}\label{prop:loc_to_formal}
  A localisable monad $T$ on a stiff category $\cat{C}$ induces a formal monad on $\overline{\cat{C}}$ in $[\ZI(\cat{C})\op,\cat{Cat}]$.
  The natural transformation $\overline{T} \colon \overline{\cat{C}} \Rightarrow \overline{\cat{C}}$ has components $T\restrict{u}$, the modification $\overline{\eta} \colon \overline{\cat{C}} \Rrightarrow \overline{T}$ has components $\eta\restrict{u}$, and the modification $\overline{\mu} \colon \overline{T}^2 \Rrightarrow \overline{T}$ has components $\mu\restrict{u}$ as in Proposition~\ref{prop:smallmonad}.
\end{proposition}
\begin{proof}
	This proof consist in verifying the naturality of $\overline{T}$, in showing that $\overline{\eta}$ and $\overline{\mu}$ are modifications (which follows directly from Lemma \ref{lem:monadmorphism}) and natural, and in proving that $\overline{\eta}$ and $\overline{\mu}$ satisfy the monad laws (which pointwise follows from Proposition~\ref{prop:smallmonad}). The complete proof is included in Appendix~\ref{sec:proofs}.
\end{proof}

\begin{theorem}
	For a stiff monoidal category $\cat{C}$ there is a bijective correspondence between localisable monads on $\cat{C}$ and formal monads on $\formal{\cat{C}}$ in $[\ZI(\cat{C})\op,\cat{Cat}]$ (via the constructions of Propositions \ref{thm:formal_to_loc} and \ref{prop:loc_to_formal}).
\end{theorem}
\begin{proof}
  Start with a localisable monad $T$ and follow Proposition \ref{prop:loc_to_formal} to get a formal monad $\formal{T}$. Then apply Proposition \ref{thm:formal_to_loc} to get a localisable monad $T'$ which we claim equals the original monad $T$. It is clear that $T'$ equals $T$ as a functor.
  It remains to check that the strength obtained this way on $T'$ is the same as the original strength on $T$. To do this, note that the strength~\eqref{eq:def_strength} from Proposition~\ref{thm:formal_to_loc} can be rewritten as follows, where $\str$ denotes the original strength from the localisable monad:
  \[
  	\varepsilon_{T\restrict{1}FGA}\circ F T\restrict{u}\eta^u_{GA} 
    = (T(A\otimes U)\otimes u)\circ (\str_{A,U}\otimes U) \otimes (T(A)\otimes (U\otimes u)^{-1})
    = \str_{A,U}
  \]
  Here we use the naturality of the strength and the fact that $U \otimes u$ is an isomorphism.
  We prove similarly that using Proposition \ref{prop:loc_to_formal} then Proposition \ref{thm:formal_to_loc} gives us back the unit and the multiplication of the starting localisable monad. To simplify the notation we used $F$ and $G$ to denote $\Fu$ and $\Gu$. 
  
  Now start with a formal monad $\formal{T}$, turn it into a localisable monad $(\formal{T}\restrict{1},\str)$, and then into a formal monad $\widetilde{T}$. 
  Then $\widetilde{T}\restrict{u}(A)=\formal{T}\restrict{u}(A)$ and $\widetilde{T}\restrict{u}$ sends a morphism $f: GA \to GB$ in $\cat{C}\restrict{u}$ given by $f \colon A \otimes U \to B$ to the morphism $\formal{T}\restrict{u}(A) \to \formal{T}\restrict{u}(B)$ in $\cat{C}\restrict{u}$ given by:
  \[
    \formal{T}\restrict{1}(A) \otimes U
    \xrightarrow{\str_{A,U}}
    \formal{T}\restrict{1}(A \otimes U)
    \xrightarrow{\formal{T}\restrict{1}(f)}
    \formal{T}\restrict{1}(B)
  \]
  We have to prove that this equals $\formal{T}\restrict{u}(f)$. To see this, first note that by the properties of the adjunction, a map $f$ in the coKleisi category $\cC\restrict{u}$ is defined in the base category as $\varepsilon \circ F(f)$, which we will denote $f^\cC$. 
   %
   With this notation, and again using $F$ and $G$ to denote $\Fu$ and $\Gu$, we get:
   \begin{align}
	\formal{T}\restrict{1}(f^\cC) \circ \str_{A,U} & = \formal{T}\restrict{1}\varepsilon_B \circ \formal{T}\restrict{1} Ff\circ \varepsilon_{\formal{T}\restrict{1} FGA} \circ F\formal{T}\restrict{u}\eta_{GA} \label{eq:pf:def1}\\
	& = \varepsilon_{\formal{T}\restrict{1} B} \circ FG\formal{T}\restrict{1} \varepsilon_B \circ FG\formal{T}\restrict{1}Ff \circ F\formal{T}\restrict{u}\eta_{GA}\label{eq:pf:nat_eps}\\
	& = \varepsilon_{\formal{T}\restrict{1} B} \circ F\formal{T}\restrict{u}G \varepsilon_B \circ F\formal{T}\restrict{u} GFf \circ F\formal{T}\restrict{u}\eta_{GA}\label{eq:pf:nat_on_mor}\\
	& = \varepsilon_{\formal{T}\restrict{1} B} \circ F\formal{T}\restrict{u}G \varepsilon_B \circ F\formal{T}\restrict{u}\eta_{GB} \circ F\formal{T}\restrict{u} f  \label{eq:pf:nat_eta}\\
	& = \varepsilon_{\formal{T}\restrict{1} B} \circ F\formal{T}\restrict{u} f \label{eq:pf:adjunction}\\
	& = (\formal{T}\restrict{u}(f))^\cC \label{eq:pf:def2}
\end{align}
Line \eqref{eq:pf:def1} follows from the definition of the strength given in Equation \eqref{eq:def_strength} and the definition of $f^\cC$. The next three lines follow from naturality of $\varepsilon$ used twice, Equation \eqref{eq:nat_on_mor}, and naturality of $\eta$ respectively. Line \eqref{eq:pf:adjunction} uses the property of the adjunction and the last line uses the definition of $(\formal{T}\restrict{u}(f))^\cC$.

  Similarly, using Proposition~\ref{thm:formal_to_loc} and then Proposition~\ref{prop:loc_to_formal} gives back the unit and the multiplication of the original formal monad.  
\end{proof}

\subsection{Graded monads and indexed monads}

We now connect to the pre-existing notions of $\cat{E}$-indexed monads and $\cat{E}$-graded monads for a monoidal category $\cat{E}$. 
Recall that an $\cat{E}$-\emph{graded} monad is a lax monoidal functor $\cat{E} \to [\cat{C},\cat{C}]$. It consists of functors $T_u \colon \cat{C} \to \cat{C}$, a natural transformation $\eta_A \colon A \to T_I(A)$, and a transformation $\mu_{u,v,A} \colon T_u(T_v(A)) \to T_{u \otimes v}(A)$ natural in $u,v$, and $A$, satisfying some coherence diagrams~\cite{fujiikatsumatamellies:gradedmonads}.

On the other hand, an $\cat{E}$-\emph{indexed} monad is a functor $\cat{E} \to \cat{Monad}(\cat{C})$. It also consists of functors $T_u \colon \cat{C} \to \cat{C}$, but now with transformations $\eta_{u,A} \colon A \to T_u(A)$ and transformations $\mu_{u,A} \colon T_u^2(A) \to T_u(A)$ natural in $u$ and $A$, such that each $(T_u,\eta_u,\mu_u)$ forms a monad. The formal monads on $\formal{\cat{C}}$ as defined in Section~\ref{sec:formalmonads} are $\ZI(\cC)$-indexed monads. The next lemma provides conditions under which indexed monads induce graded monads and vice versa.

Recall that a monoidal category \emph{has codiagonals} when there is a natural transformation $A \otimes A \to A$ that respects the coherence isomorphisms~\cite{jacobs:weakeningcontraction}.

\begin{lemma}\label{lem:graded_implies_indexed}
  Let $\cat{E}$ be a monoidal category.  
  If the tensor unit is initial, then an $\cat{E}$-indexed monad induces a $\cat{E}$-graded monad.
  If the tensor product has codiagonals, then an $\cat{E}$-graded monad induces an $\cat{E}$-indexed monad.
  If $\cat{E}$ is cocartesian, there is a bijective correspondence between $\cat{E}$-graded monads and $\cat{E}$-indexed monads.
\end{lemma}
\begin{proof}
  Suppose the tensor unit $0$ in $\cat{E}$ is initial.
  An $\cat{E}$-indexed monad $(T_u,\eta_u,\mu_u)$ then induces an $\cat{E}$-graded monad with the same $T_u$ but $\overline{\eta}_A = \eta_{0,A}$ and $\overline{\mu}_{u,v,A}$ given by:
  \[\begin{tikzcd}[column sep=18mm]
  	{T_u(T_v(A))} & {T_{u\otimes 0}(T_{0 \otimes v}(A))} & {T_{u \otimes v}^2(A)} & {T_{u\otimes v}(A)}
  	\arrow["{T_{u \otimes !}(T_{! \otimes v}(A))}", from=1-2, to=1-3]
  	\arrow["{\mu_{u \otimes v,A}}", from=1-3, to=1-4]
  	\arrow["{T_{\rho^{-1}}(T_{\lambda^{-1}}(A))}", from=1-1, to=1-2]
  \end{tikzcd}\]

  Now suppose that $\cat{E}$ has codiagonals.
  An $\cat{E}$-graded monad $(T_u,\eta,\mu_{u,v})$ then induces an $\cat{E}$-indexed monad with the same $T_u$ but $\overline{\eta}_{u,A} = \eta_A$ and $\overline{\mu}_{u,A}$ given by:
  \[\begin{tikzcd}[column sep=15mm]
  	{T_u^2(A)} & {T_{u \otimes u}(A)} & {T_u(A)}
  	\arrow["{\mu_{u,u,A}}", from=1-1, to=1-2]
  	\arrow["{T_{\nabla_u}(A)}", from=1-2, to=1-3]
  \end{tikzcd}\]

  If $\cat{E}$ is cocartesian, these two constructions are each other's inverse. For example, $\overline{\overline{\mu}}_{u,A} = \mu_{u,A}$ because:
  \[\begin{tikzcd}[column sep=18mm]
  	{T_{u+0}(T_{0+u}(A))} & {T_{u+u}^2(A)} & {T_{u+u}(A)} \\
  	& {T_u^2(A)} & {T_u(A)}
  	\arrow["{T_{\nabla_u}(A)}", from=1-3, to=2-3]
  	\arrow["{\mu_{u,A}}"', from=2-2, to=2-3]
  	\arrow["{\mu_{u+u,A}}", from=1-2, to=1-3]
  	\arrow["{T^2_{\nabla_u}(A)}", from=1-2, to=2-2]
  	\arrow["{T_{u+!}(T_{!+u}(A))}", from=1-1, to=1-2]
  	\arrow["{T_\rho(T_\lambda(A))}"', from=1-1, to=2-2]
  \end{tikzcd}\]
  Also $\overline{\overline{\eta}}_A = \eta_A$ because $! \colon 0 \to 0$ is the identity. 
  The other properties follow from naturality in $u$ and $v$.
\end{proof}
In particular, it follows that there is no difference between graded monads and indexed monads over (join-)semilattices.

\section{Examples}\label{sec:examples}

In this section we discuss three extended examples, showing that localisable monads may interpret central idempotents as locations in a computer memory (Subsection~\ref{subsec:quantum}), physical locations in a network of interacting agents (Subsection~\ref{subsec:concurrent}), or time in extended processes (Subsection~\ref{subsec:stochastic}). These examples use the following characterisation of central idempotents in functor categories.
\begin{lemma}\label{lem:centralidempotentsfunctorcategory}
  If $\cat{C}$ is a category and $\cat{D}$ is a symmetric monoidal category, then the functor category $[\cat{C},\cat{D}]$ is again symmetric monoidal under pointwise tensor products.
  Regarding $\ZI(\cat{D})$ as a full subcategory of the slice category $\cat{D}/I$, there is an isomorphism of categories:
  \[
    \ZI[\cat{C},\cat{D}] \simeq [\cat{C},\ZI(\cat{D})]
  \]
\end{lemma}
\begin{proof}
  Let $u \colon U \to I$ be a central idempotent in $[\cat{C},\cat{D}]$.
  The functor $[\cat{C},\cat{D}] \to \cat{D}$ that evaluates at a fixed object $C \in \cat{C}$ is strong monoidal and so preserves central idempotents. Hence each component $u_C \colon U(C) \to I$ represents a central idempotent in $\cat{D}$. This is functorial and gives one direction of the isomorphism.

  Conversely, let $F \colon \cat{C} \to \ZI(\cat{D})$ be a functor. Define $U \colon \cat{C} \to \cat{D}$ by $U(C)=\mathrm{dom}(F(C))$ and $u \colon U \Rightarrow I$ by $u_C = F(C)$. This is functorial and gives the other direction of the isomorphism. It is clear that these two assignments are inverses.
\end{proof}

\subsection{Quantum buffer}\label{subsec:quantum}

The (global) state monad on $\cat{Set}$ is a well-known monad that combines the properties of the reader and writer monads to implement computational side-effect in functional programming. 
It is defined as $T(-)=S\multimap (- \times S)$ for a state object $S \in \cat{Set}$. For example, to store one bit, take $S=\{0,1\}$.
The central idempotents of $\cat{Set}$ are (represented by) the empty set $\emptyset$ and the singleton set $1$. It follows that the (global) state monad is trivially localisable. This example is trivial but can be expanded in several ways:
\begin{enumerate}
	\item 
	Expanded to the category $\cat{Set}^n$, whose objects are $n$-tuples of sets and morphisms are $n$-tuples of functions. 
	The state monad on some object $A = (A_1,\ldots,A_n)$ in $\cat{Set}^n$ is 
	\[
	  T(A_1,\ldots,A_n)=(S_1,\ldots,S_n)\multimap \big((A_1,\ldots,A_n) \times (S_1,\ldots,S_n))
	\]
	for a chosen state object $S = (S_1,\ldots,S_n) \in \cat{Set^n}$. For example, to store $n$ bits, take $S_1=\cdots=S_n=\{0,1\}$.
	It follows from Lemma~\ref{lem:centralidempotentsfunctorcategory} that $\ZI(\cat{Set}^n) \simeq 2^n$. 
  While $\cat{Set}^n$ is symmetric monoidal closed, the state monad does not satisfy $T(A\multimap U) = T(A)\multimap T(U)$ as in Example~\ref{ex:strong}. There is still a strength, by currying the evaluation:
  \[
    T(A_1,\ldots,A_n)\times (U_1,\ldots,U_n) \times (S_1,\ldots,S_n)\to (S_1,\ldots,S_n)\times (A_1,\ldots,A_n) \times (U_1,\ldots,U_n)
  \]
	We have not discussed commutativity yet, but note that this strength is commutative in a sense made clear in Definition~\ref{def:commutative} below. Conceptually, this means that the computational side-effects modelled by a state monad ``over'' a region $(U_1,\ldots,U_n)$ are independent of those modelled by $(V_1,\ldots,V_n)$, assuming that $(U_1,\ldots,U_n)\times (V_1,\ldots,V_n)= 0$.

	\item 
	The localisable state monad of the previous point does not just work for cartesian closed categories such as $\cat{Set}^n$, but also for exponentiable objects in a symmetric monoidal category. For example, we can replicate it in the category $\cat{Hilb}$ of Hilbert spaces and completely positive linear maps used in quantum computation~\cite{heunenvicary:cqm}. To store one qubit, take $S=\mathbb{C}^2$. The monad then becomes $T(-) = S^* \otimes - \otimes S$, where $S^*=\cat{Hilb}(S,\mathbb{C})$ is the dual Hilbert space, which is isomorphic to $T(A) = A \otimes \mathbb{M}_2$, where $\mathbb{M}_2$ is the Hilbert space of complex 2-by-2 matrices.
	Similarly, to store $n$ qubits, move to $\cat{Hilb}^n$. 
	We can now see a phenomenon that didn't occur for cartesian categories: rather than a quantum memory, this monad models a quantum buffer of $n$ qubits, because there is no entanglement between the different qubits. 
	Because $\ZI(\cat{Hilb})=\{0,\mathbb{C}\}$, again $\ZI(\cat{Hilb}^n) \simeq 2^n$. The strength map is yet again given by the curry of the evaluation map, which makes $T(-)$ a commutative localisable monad in the sense of Definition~\ref{def:commutative} below.

	\item
	We can also promote the (global) state monad on $\cat{Set}$ in another direction, namely from $n=1$ or finite $n$ to an arbitrary topological space $X$ indexing the bits to be stored. Consider the category $\mathrm{Sh}(X)$ of ($\cat{Set}$-valued) sheaves on $X$, take $S$ to be the constant sheaf $S(U)=\{0,1\}$, and define $T(-) = S \multimap (- \otimes S)$. As in Example~\ref{ex:cartesian}, the central idempotents correspond to open subsets $U \subseteq X$, and this monad is still localisable. Its stalks (as discussed in Remark~\ref{rem:stalks}) are the simple (global) state monads on $\cat{Set}$ storing a single bit each.

	\item
	Points 2 and 3 combine to model a quantum buffer over an arbitrary locally compact Hausdorff topological space $X$. Consider the category $\cat{Hilb}_{C_0(X)}$ of Hilbert modules over $C_0(X)$, take $S$ to be Hilbert module $C_0(X,\mathbb{C}^2)$ of continuous functions $X \to \mathbb{C}^2$ that vanish at infinity, and define $T(-) = S^* \otimes - \otimes S$. As in Example~\ref{ex:hilbertmodules}, central idempotents are open subsets $U \subseteq X$. Again, this monad is localisable, with $T\restrict{U} = S_u^* \otimes - \otimes S_u$ for $S_u = C_0(U,\mathbb{C}^2)$. In fact, this example is related to the one in point 3, as Hilbert modules over $C_0(X)$ correspond to a Hilbert space internal to the topos $\mathrm{Sh}(X)$ by Takahashi's Theorem~\cite{barbosaheunen:sheaves,heunenreyes:frobenius}.
\end{enumerate}

\subsection{Concurrent processes}\label{subsec:concurrent}

Suppose $M_1$ is a monoid of actions that some agent 1 can perform, and $M_2$ is a monoid of actions that an agent 2 can perform. They could, for example, be free monoids over sets of atomic actions. Then we can form the coproduct $M_1+M_2$ of monoids, and quotient out a congruence that specifies $ab=ba$ for $a \in M_1$ and $b \in M_2$ when actions $a$ and $b$ are independent, to get the monoid $M$ of Mazurkiewicz traces~\cite{diekertmetivier:traces,winskelnielsen:concurrency}. Now $M$ localises to $M_1$ by projections $M \to M_i$ that disregard actions of the other agent.

The following lemma engineers a single category with two central idempotents and a monoid, that localises to the given ones. The idea is to take a product of categories, but to add \emph{silent actions}, that enforce the order in which both agents' actions occur, as in the pi calculus~\cite{milner:picalculus}.

\begin{lemma}\label{lem:bottomup}
  Let $M_1$ and $M_2$ be monoids in symmetric monoidal categories $\cat{C}_1$ and $\cat{C}_2$ that have an initial object $0$ satisfying $A \otimes 0\simeq 0$ for all objects $A$. There is a symmetric monoidal category $\cat{C}$ with a monoid $M$ and central idempotents $u_1, u_2$, that allows an isomorphism $\cat{C}\restrict{u_i} \simeq \cat{C}_i$ of monoidal categories under which $M_i$ corresponds with $\cat{C}\restrict{u_i \leq 1}(M)$.
\end{lemma}
If $\cat{C}_i$ does not yet have an initial object $0$ satisfying $A \otimes 0 \simeq 0$, we may freely adjoin one to obtain a well-defined symmetric monoidal category.
\begin{proof}
  First construct a new category $\cat{C}'$. 
	Objects are pairs $(A,B)$ of $A \in \cat{C}_1$ and $B \in \cat{C}_2$.
	Morphisms $(A,B) \to (A',B')$ include pairs $(f,g)$ of $f \in \cat{C}_1(A,A')$ and $g \in \cat{C}_2(B,B')$, to which we freely adjoin morphisms $\tau_{A,B} \colon (A,B) \to (A,B)$ for each object $(A,B)$.
  Thus morphisms are finite lists $\big( (f_1,g_1), \tau_1, \ldots, \tau_{n-1}, (f_n,g_n) \big)$ where the domain of $\tau_n$ is the codomain of $f_n \otimes g_n$. Composition concatenates and then contracts:
  \begin{align*}
    & \big( (f'_1,g'_1), \tau'_1, \ldots, (f'_n,g'_n) \big) 
    \circ
    \big( (f_1,g_1), \tau_1, \ldots, (f_m,g_m) \big)
    \\&= 
    \big( (f_1,g_1), \tau_1, \ldots, (f'_1 \circ f_m, g'_1 \circ g_m), \tau, \ldots, (f'_n,g'_n) \big)
  \end{align*}
  Defining identity to be the trivial list $(\mathrm{id}[A],\mathrm{id}[B])$ makes $\cat{C}'$ into a well-defined category.

  Next, take the free symmetric monoidal category $\cat{C}''$ on $\cat{C}'$. Objects of $\cat{C}''$ are finite lists of objects of $\cat{C}'$, and morphisms are pairs $(\pi, h_1,\ldots,h_n)$ of a permutation $\pi$ of list indices and a list of morphisms in $\cat{C}'$; see for example~\cite{abramsky:scalars}.
  Finally, consider the generalised equivalence relation~\cite{bednarczyketal:congruences} on $\cat{C}''$ generated by 
  \begin{align*}
    (I,0) \otimes \tau_{A,B} & \sim (A,0) \\
    (0,I) \otimes \tau_{A,B} & \sim (0,B) \\
    \big( \pi, (f_1,g_1),(f_2,g_2)\big) & \sim 
    (\sigma,(f_1 \otimes f_2, g_1 \otimes g_2)) 
  \end{align*}
  where $\pi$ is the bijection $1 \mapsto 2$ and $2 \mapsto 1$ on $\{1,2\}$. This is a symmetric monoidal congruence, so $\cat{C} = \cat{C}'' \slash \mathop{\sim}$ is a well-defined symmetric monoidal category.

  Because $0$ is initial and $A \otimes 0 = 0$ in $\cat{C}_i$, the objects $(I,0)$ and $(0,I)$ in $\cat{C}'$ become central idempotents $u_1,u_2$ in $\cat{C}$, and moreover $(A,B) \mapsto [A]_\sim$ is an isomorphism $\cat{C}\restrict{u_1} \simeq \cat{C}_1$ and similarly for $u_2$. 
  Finally, $M=(M_1,M_2)$ is a monoid in $\cat{C}$, that localises to $M_i$ by construction.
\end{proof}

In the proof of the previous lemma, we could alternatively have described $\cat{C}$ as consisting of formal string diagrams generated by $\cat{C}_1 \times \cat{C}_2$ and the silent actions $\tau_{A,B}$~\cite{curienmimram:monoidalpresentations}, or as terms in a formal syntactic language~\cite{jay:monoidallanguages}.

\begin{example}
  Let $M_i$ be monoids in $\cat{C}_i=\cat{Set}$. 
  They induce writer monads $T_i(A)=M_i \otimes A$ on $\cat{C}_i$.
  Now the monoid $M$ in the category $\cat{C}$ of the previous lemma induces a writer monad $T$ on $\cat{C}$.
  The monad $T$ is localisable by Example~\ref{ex:strong}, and $T\restrict{i}$ corresponds to $T_i$ under the isomomorphism $\cat{C}\restrict{i} \simeq \cat{C}_i$.
  Thus $T$ tracks the agents' actions as side effects during a (distributed) computation. 
\end{example}

It seems possible to extend this example to a network where the communicating agents form the points of an arbitrary topological space.

\subsection{Stochastic processes}\label{subsec:stochastic}

Write $\cat{Meas}$ for the category of measurable spaces and measurable functions. This is a symmetric monoidal category, where the tensor unit is the singleton set with its unique $\sigma$-algebra, and the tensor product of two measurable spaces is the cartesian product of the sets with the tensor product of the $\sigma$-algebras. The monoidal category $\cat{Meas}$ has only two central idempotents: the empty set $\emptyset$, and the tensor unit $1$ itself. 

Instead, consider the functor category $[\mathbb{N},\cat{Meas}]$, where the partially ordered set $\mathbb{N}$ is considered as a category by having a morphism $m \to n$ if and only if $m \leq n$. Its objects are sequences $X_1,X_2,X_3,\ldots$ of measurable spaces. Lemma \ref{lem:centralidempotentsfunctorcategory} shows that this category has many more central idempotents.
It follows that central idempotents $u \colon U \Rightarrow 1$ in $[\mathbb{N},\cat{Meas}]$ correspond to upward-closed subsets of $\mathbb{N} \cup \{\infty\}$, or more succinctly, to elements of $n \in \mathbb{N} \cup \{\infty\}$, by
\[
  U(m) = \begin{cases}
    \emptyset & \text{ if } m<n \\
    1 & \text{ if } m \geq n
  \end{cases}
\]

The \emph{Giry monad} $G \colon \cat{Meas} \to \cat{Meas}$ takes a measurable space to the set of probability measures on it~\cite{giry:girymonad}. It extends to a monad on $[\cat{N},\cat{Meas}]$. 

\begin{example}
  The monad $\widehat{G} = G \circ (-)$ on $[\cat{N},\cat{Meas}]$ is localisable, where the maps $\widehat{G}(X) \otimes U \Rightarrow \widehat{G}(X \otimes U)$ can simply be taken to be identities (because $G(\emptyset)=\emptyset$).
  The restricted category $[\cat{N},\cat{Meas}]\restrict{n}$ is $[\{n,n+1,\ldots\},\cat{Meas}]$, and the monad $\widehat{G}\restrict{n}$ is simply the restriction of $\widehat{G}$ to $\{n,n+1,\ldots\}$.

  The adjunction between $\cat{Meas}$ and the Kleisli category $\Kl(G)$ lifts to an adjunction between $[\cat{N},\cat{Meas}]$ and $[\cat{N},\Kl(G)]$. The latter is not equivalent to the Kleisli category of $\widehat{\cat{G}}$ because the functor $[\cat{N},\cat{Meas}] \to [\cat{N},\Kl(G)]$ that turns a sequence of elements of measurable spaces into a sequence of Dirac measures it not essentially surjective~\cite[Theorem~9]{westerbaan:kleisli}. 

  The objects of $[\cat{N},\cat{Meas}]$ are \emph{stochastic processes}~\cite{lawvere:probabilistic,giry:girymonad,fritz:markov}. Instead of $(\mathbb{N},\leq)$, we could equally well have taken continuous time $(\mathbb{R}^{\geq 0},\leq)$. In fact, we could also have regarded the monoid $(\mathbb{N},+,0)$ or $(\mathbb{R}^{\geq 0},+,0)$ as a one-object category. Then $[\cat{N},\Kl(G)]$ would consist of stationary processes, but the central idempotents would remain the same by Lemma~\ref{lem:centralidempotentsfunctorcategory}: ideals of $\mathbb{N}$ or $\mathbb{R}^{\geq 0}$ under $+$ are also upward-closed subsets.
\end{example}

Rather than stochastic (Markov) processes, that depend on the history thus far (one time step ago only), we could have taken more interesting partially ordered sets than the totally ordered ones $\mathbb{N}$ and $\mathbb{R}^{\geq 0}$. 

\section{Algebras}\label{sec:algebras}

Let $\cat{C}$ be a symmetric monoidal category. As we have seen in Section~\ref{sec:formalmonads}, a localisable monad $T \colon \cat{C} \to \cat{C}$ is equivalently described as a formal monad $T\restrict{-}$ in the 2-category $\cat{K}=[\ZI(\cat{C})\op,\cat{Cat}]$. 
What are its formal (Eilenberg-Moore) algebras?

The general answer is described in~\cite{lackstreet:formalmonads2,street:formalmonads}. The formal algebra category is an object of $\cat{K}$ satisfying the following. 
For any object $X \in \cat{K}$, the formal monad $T\restrict{-}$ induces a (concrete) monad $K(X,T\restrict{-})$ on the category $\cat{K}(X,\cat{C}\restrict{-})$; this monad sends a natural transformation $\beta \colon X \Rightarrow \cat{C}\restrict{-}$ to the natural transformation with components $T\restrict{u} \circ \beta_u \colon X_u \to \cat{C}\restrict{u}$. This (concrete) monad has a (concrete) Eilenberg-Moore category of algebras. Objects are pairs of a natural transformation $\beta$ and a modification $\theta$ of type
\begin{equation}\label{eq:formalalgebra}\begin{tikzcd}
	& {\cat{C}\restrict{u}} \\
	{X_u} & {} & {\cat{C}\restrict{u}}
	\arrow["{\beta_u}", curve={height=-15pt}, from=2-1, to=1-2]
	\arrow["{T\restrict{u}}", curve={height=-15pt}, from=1-2, to=2-3]
	\arrow["{\beta_u}"', from=2-1, to=2-3]
	\arrow["{\theta_u}"', shorten <=6pt, shorten >=6pt, Rightarrow, from=1-2, to=2-2]
\end{tikzcd}\end{equation}
satisfying the algebra laws.
Morphisms are modifications $\varphi \colon \beta \Rrightarrow \beta'$ satisfying:
\begin{equation}\label{eq:formalalgebramap}\begin{tikzcd}[column sep=4mm]
	& {X_u} \\
	{\cat{C}\restrict{u}} \\
	& {\cat{C}\restrict{u}}
	\arrow["{\beta_u}"', curve={height=12pt}, from=1-2, to=2-1]
	\arrow["{T\restrict{u}}"', curve={height=12pt}, from=2-1, to=3-2]
	\arrow[""{name=0, anchor=center, inner sep=0}, "{\beta_u}"{description}, from=1-2, to=3-2]
	\arrow[""{name=1, anchor=center, inner sep=0}, "{\beta'_u}", out=-10,in=10,looseness=1.3, from=1-2, to=3-2]
	\arrow["{\theta_u}", shorten <=6pt, shorten >=9pt, Rightarrow, from=2-1, to=0]
	\arrow["{\varphi_u}", shorten <=12pt, shorten >=9pt, Rightarrow, from=0, to=1]
\end{tikzcd}
\qquad = \qquad
\begin{tikzcd}[column sep=4mm]
	& {X_u} \\
	{\cat{C}\restrict{u}} \\
	& {\cat{C}\restrict{u}}
	\arrow[""{name=0, anchor=center, inner sep=0}, "{\beta_u}"', curve={height=12pt}, from=1-2, to=2-1]
	\arrow["{T\restrict{u}}"', curve={height=12pt}, from=2-1, to=3-2]
	\arrow[""{name=1, anchor=center, inner sep=0}, "{\beta'_u}"{description}, curve={height=-25pt}, from=1-2, to=2-1]
	\arrow[""{name=2, anchor=center, inner sep=0}, "{\beta'_u}", out=-10,in=10,looseness=1.3, from=1-2, to=3-2]
	\arrow["{\varphi_u}", shorten <=7pt, shorten >=7pt, Rightarrow, from=0, to=1]
	\arrow["{\theta'_u}"', shorten <=15pt, shorten >=15pt, Rightarrow, from=1, to=2]
\end{tikzcd}\end{equation}
This defines the object-part of a 2-functor $\cat{K}\op \to \cat{Cat}$. Now $A \in \cat{K}$ is the \emph{formal algebra} object of the formal monad $\cat{T}\restrict{-}$ when this 2-functor is naturally isomorphic to $\cat{K}(-,A)$.

\begin{proposition}\label{prop:algebras}
  Let $T$ be a localisable monad on a symmetric monoidal category $\cat{C}$.
  The formal monad $T\restrict{-}$ in $[\ZI(\cat{C})\op,\cat{Cat}]$	has a formal algebra object $A\restrict{-}$ where $A\restrict{u}=\Alg(T\restrict{u})$ is the category of algebras of $T\restrict{u}$.
\end{proposition}
\begin{proof}
  If $u \leq v$	then the monad morphism $\cat{C}\restrict{u \leq v}$ of Lemma~\ref{lem:monadmorphism} induces a functor $A\restrict{v} \to A \restrict{u}$, so $A$ is a well-defined object of $\cat{K}=[\ZI(\cat{C})\op,\cat{Cat}]$.
  Now, for an object $X \in \cat{K}$, the hom-category $\cat{K}(X,A)$ has as objects natural transformations $\beta_u \colon X_u \to \Alg(T\restrict{u})$. But the objects of $\Alg(T\restrict{u})$ are themselves morphisms $\theta_u \colon T\restrict{u}(B) \to B$ in $\cat{C}\restrict{u}$, that furthermore satisfy the algebra laws. These assemble into a modification satisfying~\eqref{eq:formalalgebra}. It is labour-intensive but straightforward to verify that the morphisms of $\Alg(T\restrict{u})$ similarly match modifications satisfying~\eqref{eq:formalalgebramap}, and that this in fact gives a 2-natural isomorphism to $A\restrict{-}$.
  Thus $A\restrict{-}$ is a formal algebra object.
\end{proof}

Similarly, a \emph{formal Kleisli algebra} object of the formal monad $T\restrict{-}$ is characterised in~\cite{lackstreet:formalmonads2,street:formalmonads} as a formal algebra object in the 2-category $[\ZI(\cat{C})\op,\cat{Cat}\op]$, where $\cat{Cat}\op$ has reversed the 1-cells but not the 2-cells of $\cat{Cat}$.

\begin{corollary}\label{cor:kleisli}
  Let $T$ be a localisable monad on a symmetric monoidal category $\cat{C}$. The formal monad $T\restrict{-}$ in $[\ZI(\cat{C})\op,\cat{Cat}]$	has a formal Kleisli object $K\restrict{-}$ where $K\restrict{u}=\Kl(T\restrict{u})$ is the Kleisli category of $T\restrict{u}$.
  \qed
\end{corollary}

A Kleisli category of a commutative monad on a symmetric monoidal category is again symmetric monoidal~\cite{day:kleisli}. It would be interesting to see if there is a notion that stands to localisability as commutativity stands to strength, that guarantees that the formal Kleisli algebra object of the previous corollary is a monoid in $\cat{K}=[\ZI(\cat{C})\op,\cat{Cat}]$. We leave this for future work, but give a tentative (re)definition now.

\begin{definition}\label{def:commutative}
  A localisable monad $T$ on a symmetric monoidal category $\cat{C}$ is \emph{commutative} when:
  \begin{equation}
\begin{tikzcd}[row sep = 3.2mm]
	{T(A) \otimes U \otimes V} && {T(A \otimes U) \otimes V} && {T(A \otimes U \otimes V)} \\
	\\
	{T(A) \otimes V \otimes U} && {T(A \otimes V) \otimes U} && {T(A \otimes V \otimes U)}
	\arrow["{T(A) \otimes \sigma_{U,V}}"', from=1-1, to=3-1]
	\arrow["{\str_{A,V}\otimes U}"', from=3-1, to=3-3]
	\arrow["{\str_{A \otimes V,U}}"', from=3-3, to=3-5]
	\arrow["{T(A \otimes \sigma_{V,U})}"', from=3-5, to=1-5]
	\arrow["{\str_{A,U} \otimes V}", from=1-1, to=1-3]
	\arrow["{\str_{A \otimes U,V}}", from=1-3, to=1-5]
\end{tikzcd}
  \label{eq:commutative}
  \end{equation}
\end{definition}

It follows from this definition that if $u \wedge v = 0$, then the computational side-effects modeled by $T_u$ and $T_v$ do not influence each other. Intuitively, side-effects $T_u$ and $T_v$ that act in disjoint areas must be independent of each other.

\section{Further work}\label{sec:conclusion}

There are several interesting directions for further research.
\begin{itemize}
  \item We have decomposed a localisable monad into monads on local monoidal categories, but can a monad on a local monoidal category be decomposed further? For example, the local state monad~\cite{plotkinpower:monads} is based on the presheaf category $[\cat{Inj},\cat{Set}]$. Its central idempotents correspond to natural numbers, topologised by saying that a subset is open when it is upward-closed under the usual ordering of natural numbers. This topological space is already local: every net converges to the focal point $0$. The `decomposition' using coends of~\cite{plotkinpower:monads} relies on the base category $[\cat{Inj},\cat{Set}]$ having much more structure rather than just a monoidal category. The successor function of natural numbers there affords the possibility to allocate fresh locations. Our example of local states in Section~\ref{subsec:quantum} completely ignored this possibility. Can this extra structure be axiomatised -- using open sets rather than points -- and used for a further decomposition? 

	\item Two monads on the same base category can be composed as soon as there is a distributive law between them~\cite{beck:distributivelaws,zwart:thesis}. When does a distributive law respect the localisable nature of the monads, and how does it interact with their decomposition into monads on local monoidal categories? 

	\item More generally than monads, when is a PROP localisable, and how does a localisable PROP decompose into local ones~\cite{lack:props,power:indexedlawvere}? 

	\item Formal monads form a bridge between the ``top-down'' localisable monads and the ``bottom-up'' approaches. Can this relationship be made more constructive? Given monads $T_i$ on possibly different monoidal base categories $\cat{C}_i$, can we construct a monad $T$ on a monoidal category $\cat{C}$ with central idempotents $i$ such that $\cat{C}\restrict{i} \simeq \cat{C}_i$ and $T\restrict{i} \simeq T_i$? The free construction of Lemma~\ref{lem:bottomup} is an initial step in this direction; can it be given a more elegant concrete description, and extended to arbitrary topogical spaces?

	\item Is there a notion that stands to localisability as commutativity stands to strength, that guarantees that the formal Kleisli object of Corollary~\ref{cor:kleisli} is a monoid in $[\ZI(\cat{C})\op,\cat{Cat}]$? Does it connect to partial commutativity as in the Mazurkiewicz traces of Section~\ref{subsec:concurrent}?
\end{itemize}

\bibliographystyle{plainurl}
\bibliography{bibliography}

\appendix
\section{Deferred proofs}\label{sec:proofs}

This appendix provides complete proofs of results that were shortened or omitted from the main text.

\restate{lem:monadmorphism}
\begin{lemma}
  Let $T$ be a localisable monad on $\cat{C}$.
  If $u \leq v$ are central idempotents, then the functor $\cat{C}\restrict{u \leq v}$ from Lemma~\ref{lem:adjunctionCu} is a (lax) monad morphism $T\restrict{v} \to T\restrict{u}$ with $\varphi_A = T(A) \otimes u$.
\end{lemma}
\begin{proof}
  Naturality of $\varphi$ comes down to commutativity of the following diagram in $\cat{C}$:
	\[\begin{tikzcd}[column sep=14mm]
		{T(A) \otimes U \otimes U} & {T(A) \otimes U \otimes V} & {T(A) \otimes U \otimes V} & {T(A) \otimes V} \\
		{T(A \otimes U) \otimes U} & {T(A \otimes U) \otimes V} & {T(A) \otimes V \otimes V} \\
		{T(A \otimes V) \otimes U} & {T(A \otimes V)} & {T(A \otimes V)} & {T(A \otimes V)} \\
		{T(B) \otimes U} &&& {T(B)}
		\arrow["{\str^u_{A,U} \otimes U}"', from=1-1, to=2-1]
		\arrow["{T(A \otimes m) \otimes U}"', from=2-1, to=3-1]
		\arrow["{T(f) \otimes U}"', from=3-1, to=4-1]
		\arrow["{\varphi_B = T(B) \otimes u}"', from=4-1, to=4-4]
		\arrow["{T(f)}", from=3-4, to=4-4]
		\arrow["{\str^v_{A,V}}", from=1-4, to=3-4]
		\arrow["{T(A) \otimes U \otimes m}", from=1-1, to=1-2]
		\arrow["{T(A \otimes u) \otimes m}"', from=2-1, to=2-2]
		\arrow["{\str^u_{A,U} \otimes V}"', from=1-2, to=2-2]
		\arrow["{T(A \otimes V) \otimes u}"', from=3-1, to=3-2]
		\arrow["{T(A \otimes m) \otimes v}"', from=2-2, to=3-2]
		\arrow["{\str^v_{A,V} \otimes V}", from=2-3, to=3-3]
		\arrow[Rightarrow, no head, from=3-3, to=3-4]
		\arrow[Rightarrow, no head, from=3-3, to=3-2]
		\arrow["{T(A) \otimes u \otimes V}", from=1-3, to=1-4]
		\arrow["{T(A) \otimes m \otimes V}", from=1-3, to=2-3]
		\arrow[Rightarrow, no head, from=1-2, to=1-3]
	\end{tikzcd}\]
	Here the top centre rectangle follows from~\eqref{eq:st:naturalinu}.
	The diagrams~\eqref{eq:monadmorphism} in $\cat{C}\restrict{v}$ unfold to the following in $\cat{C}$:
	\[\begin{tikzcd}[column sep=1mm]
		{A \otimes U \otimes U} && {T(A) \otimes U} \\
		{A \otimes U} && {T(A)} \\
		& {A \otimes V}
		\arrow["{\eta_A \otimes u \otimes U}", from=1-1, to=1-3]
		\arrow["{T(A) \otimes u}"', from=1-3, to=2-3]
		\arrow["{A \otimes m}"', from=2-1, to=3-2]
		\arrow["{\eta_A \otimes v}"', from=3-2, to=2-3]
		\arrow["{\eta_A \otimes u}", from=2-1, to=2-3]
		\arrow["{A \otimes U \otimes u}", from=1-1, to=2-1]
	\end{tikzcd}
	\begin{tikzcd}[column sep=2mm]
		{T^2(A) \otimes U \otimes U} && {T^2(A) \otimes U \otimes U \otimes U} \\
		{T(A) \otimes U} && {T(T(A) \otimes U)\otimes U \otimes U} \\
		&& {T^2(A) \otimes U \otimes U} \\
		{T(A)} && {T^2(A) \otimes U} \\
		& {T^2(A) \otimes V}
		\arrow["{(T^2(A) \otimes U \otimes U \otimes u)^{-1}}", from=1-1, to=1-3]
		\arrow["{\mu_A \otimes v}", from=5-2, to=4-1]
		\arrow["{\mu_A\otimes u \otimes U}", from=1-1, to=2-1]
		\arrow["{\str_{T^2(A),U}\otimes U \otimes U}"', from=1-3, to=2-3]
		\arrow["{T^2(A) \otimes u \otimes U}"', from=3-3, to=4-3]
		\arrow["{T^2(A) \otimes m}", from=4-3, to=5-2]
		\arrow["{\mu_A \otimes u}", from=4-3, to=4-1]
		\arrow["{T(A) \otimes u}", from=2-1, to=4-1]
		\arrow["{T(T(A) \otimes u) \otimes U \otimes U}"', from=2-3, to=3-3]
	\end{tikzcd}\]
	That these commute follows from~\eqref{eq:st:naturalinu} and bifunctoriality of the tensor product.
\end{proof}

For the next proof, to simplify the notation we will rename the adjoint functors $\Fuv$ and $\Guv$ into $F^{\utv}$ and $G^{\utv}$. Moreover, the monad $T\restrict{v}$ will be denoted by $T_v$.

\restate{thm:formal_to_loc}
\begin{proposition}
	Let $\cC$ be a stiff category. Let $(\formal{T},\formal{\mu},\formal{\eta})$ be a formal monad in $[\ZI(\cat{C})^{\op}, \cat{Cat}]$ above $\formal{\cC}$ and let $u \leq v$ be central idempotents. Then the monad $T_v$ is a localisable monad with the strength $\str_{A,U} \colon T_v(A)\otimes U \to T_v(A\otimes U)$ defined as the following composition in $\cC\restrict{v}$ for any object $A$ in $\cC\restrict{v}$:
\begin{equation*}
\begin{tikzcd}
	{T_v(A)\otimes U=F^\utv G^\utv T_vA = F^\utv T_u G^\utv A} \\
	{F^\utv T_u G^\utv F^\utv  G^\utv A = F^\utv  G^\utv T_v F^\utv  G^\utv A} \\
	{T_v F^\utv  G^\utv A = T_v(A\otimes U)}
	\arrow["{\varepsilon^{u\leq v}_{T_v F^\utv G^\utv  A}}", from=2-1, to=3-1]
	\arrow["{F^\utv  T_u\eta^{u\leq v}_{G^\utv  A}}", from=1-1, to=2-1]
\end{tikzcd}
\end{equation*}
where $\eta^{u\leq v}$ and $\varepsilon^{u\leq v}$ are the unit and counit of adjunction $F^\utv \dashv G^\utv$.
\end{proposition}

\begin{proof}
We need to prove each of the axioms of Definition \ref{def:localisablemonad}. This consist of many commutativity diagrams which we  present below. In order to simplify the already very heavy notation, we make the following changes. Whenever the context is clear, we will drop the superscripts $u\leq v$ and simply write $F$ and $G$ (and $\eta$ and $\varepsilon$ for the unit $\eta^{\utv}$ and counit $\varepsilon^{\utv}$ of the adjunction). We will also omit the $v=1$ when possible and simply write $F^u$ and $G^u$ for $F^{u\leq 1}$ and $G^{u\leq 1}$. Finally, we write out the proof for the case $v=1$; the same arguments hold for any $T_v$ by using the relevant strength.

\begin{enumerate}
\item Condition~\eqref{eq:st:unitor}, that $T_1(\rho_A)\circ \str_{A,I}=\rho_{TA}$, follows from the commutativity of:
 \[\begin{tikzcd}
	{FT_uGA} & {FGT_1A} && {T_1A} \\
	\\
	{FT_uGFGA} & {FGT_1FGA} && {T_1FGA}
	\arrow[""{name=0, anchor=center, inner sep=0}, "{\varepsilon_{T_1A}}", from=1-2, to=1-4]
	\arrow["{T_1\varepsilon_A}"', from=3-4, to=1-4]
	\arrow[""{name=1, anchor=center, inner sep=0}, Rightarrow, no head, from=1-2, to=1-1]
	\arrow["{FT_u\eta_{GA}}"', from=1-1, to=3-1]
	\arrow[""{name=2, anchor=center, inner sep=0}, Rightarrow, no head, from=3-1, to=3-2]
	\arrow[""{name=3, anchor=center, inner sep=0}, "{\varepsilon_{T_1FGA}}"', from=3-2, to=3-4]
	\arrow["{FGT_1\varepsilon_A}"{description}, from=3-2, to=1-2]
	\arrow[Rightarrow, draw=none, from=2, to=1]
	\arrow["{\textbf{nat. }\varepsilon}"{description}, shift right=1, Rightarrow, draw=none, from=3, to=0]
\end{tikzcd}\]
Note that $\rho_A= \varepsilon_A$. The left square follows from the zigzag equation $G\varepsilon_A\circ \eta_{GA} = \id_{GA}$.

\item Next, consider equation~\eqref{eq:st:associator}:
  $$T_1(\alpha_{A,U,V}) \circ \str_{A,U \otimes V}  = \str_{A \otimes U,V} \circ (\str_{A,U} \otimes V) \circ \alpha_{T_1A,U,V}.$$ 
  This comes down to a large commuting diagram, shown in Figure~\ref{fig:axiom2_proof} on page~\pageref{fig:axiom2_proof}.
  First note the following decomposition of the counit, that follows from Lemma~\ref{lem:decompFG}. For central idempotents $u\leq v$, we have that $F^u=F^v \circ F^{u\leq v}$ and $G^u=G^{u\leq v} \circ G^v$, and hence a map
\[\begin{tikzcd}
	{F^uG^u} && {F^vG^v}
	\arrow["{F^v\varepsilon^{u\leq v}_{G^vA}}", from=1-1, to=1-3]
\end{tikzcd}\]
where $\varepsilon^{\utv}$ if the counit of the adjunction $F^\utv\dashv G^\utv$.
Moreover, as a consequence of Lemma~\ref{lem:decompFG}:
\begin{align}\label{eq:decomp_counit}
\varepsilon^{u}=\varepsilon^{v}\circ F^v(\varepsilon^{u\leq v}_{G^v}).
\end{align}
Using this, we note that the associator $\alpha_{A,U,V}$, as an operation on the adjoint functors $F^u\dashv G^u$, is defined as the composition:
\[\begin{tikzcd}
	{F^\uov G^\uov A} && {F^uG^uA} && {F^uG^uF^uG^uA} && {F^vG^vF^uG^uA}
	\arrow["{F^u\varepsilon^{\uov\leq u}_{G^uA}}", from=1-1, to=1-3]
	\arrow["{F^v\varepsilon^{u\leq v}_{G^vF^uG^uA}}", from=1-5, to=1-7]
	\arrow["{F^u\eta_{G^uA}}", from=1-3, to=1-5]
\end{tikzcd}\]
The commutativity of Figure \ref{fig:axiom2_proof} heavily relies on the naturality of the unit and counit. It also uses the decomposition~\eqref{eq:decomp_counit} of the counit. Moreover, the entire diagram implicitly uses the decomposition of the adjoint functors presented in Lemma~\ref{lem:decompFG} and equation~\eqref{eq:equal_ob} which is a consequence of the naturality of $\formal{T}$.
Additionally, section $\textbf{(a)}$ in Figure \ref{fig:axiom2_proof} holds due to the commutativity of the following diagram:
%
\[\begin{tikzcd}
	{F^uT_uG^uF^uG^uA=F^uG^uT_1F^uG^uA} &&&& {F^u G^u T_1 F^\uov G^\uov A} \\
	\\
	\\
	{F^uG^uT_1F^uG^uF^uG^uA} \\
	\\
	{F^u G^u T_1F^vG^vF^uG^uA}
	\arrow["{F^uG^uT_1F^u\varepsilon^{\uov\leq u}_{G^uA}}"{description}, from=1-5, to=1-1]
	\arrow["{F^uG^uT_1(\alpha_{A,U,V})}"{description}, from=1-5, to=6-1]
	\arrow["{F^uG^uT_1F^v\varepsilon^{u\leq v}_{G^vF^uG^uA}}"{description}, from=4-1, to=6-1]
	\arrow["{F^uG^uT_1\varepsilon^u_{F^uG^uA}}"{description, pos=0.4}, from=4-1, to=1-1]
	\arrow["{F^uG^uT_1F^u\eta_{G^uA}}"{description, pos=0.4}, shift left=8, curve={height=-30pt}, from=1-1, to=4-1]
	\arrow["{F^uG^uT_1\varepsilon^v_{F^uG^uA}}"{description, pos=0.6}, shift left=15, curve={height=-30pt}, from=6-1, to=1-1]
\end{tikzcd}\]
This diagram uses the definition of the associator, the zigzag equation $\varepsilon_{F^uG^uA}\circ \eta_{G^uA} = \id_{G^uA}$, and the decomposition property of the counit.

 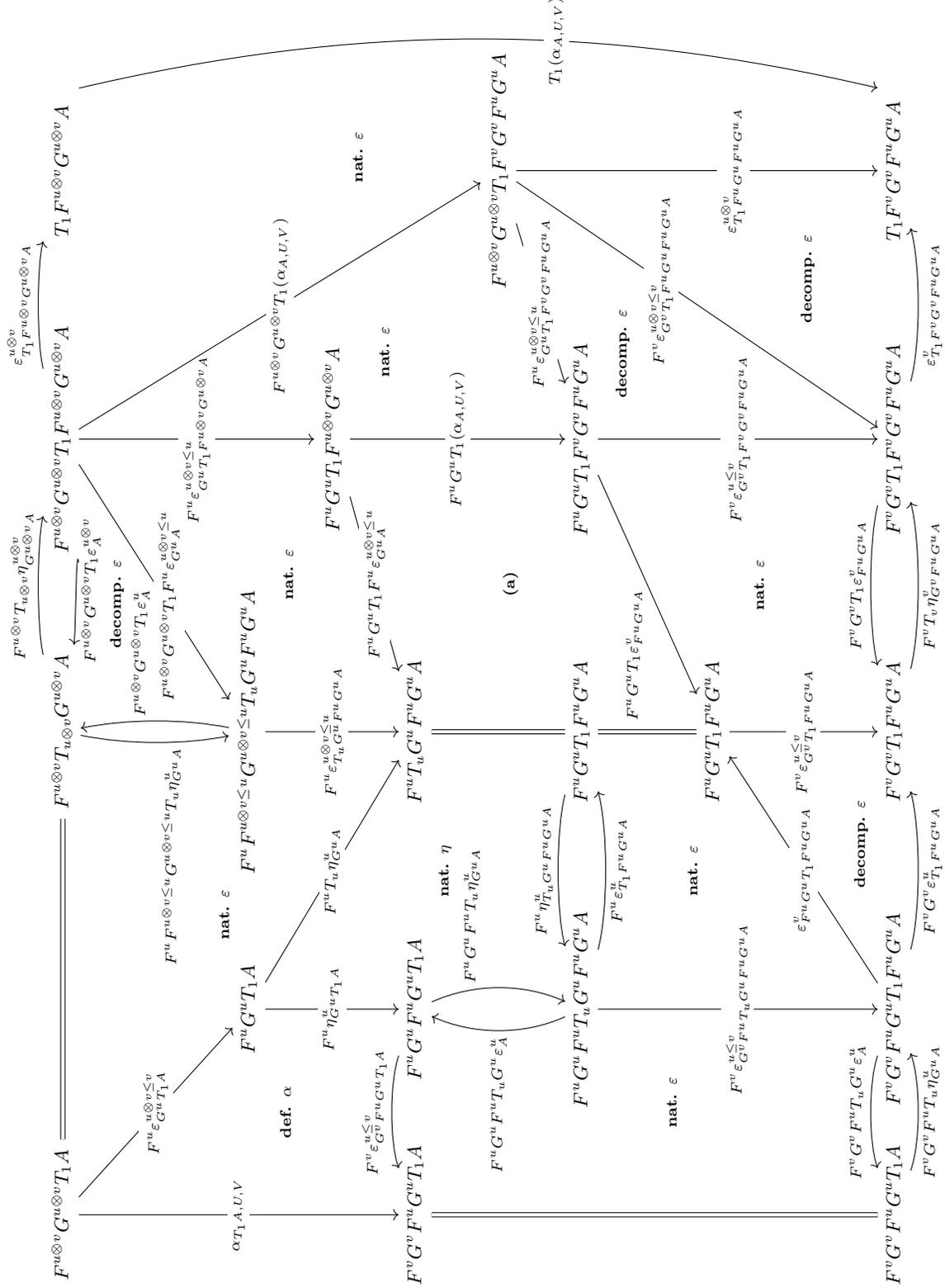
\begin{sidewaysfigure}[b!]
\[\begin{tikzcd}[column sep = 3mm, row sep = 8mm, font = \small]
	{F^\uov G^\uov T_1A} && {F^\uov T_\uov G^\uov A} & {F^\uov G^\uov T_1 F^\uov G^\uov A} & {T_1F^\uov G^\uov A} \\
	\\
	\\
	& {F^uG^uT_1A} & {F^uF^{\uov\leq u}G^{\uov\leq u}T_uG^uF^uG^uA} \\
	&&& {F^u G^u T_1 F^\uov G^\uov A} \\
	{F^vG^vF^uG^uT_1A} & {F^uG^uF^uG^uT_1A} & {F^uT_uG^uF^uG^uA} \\
	&&&& {F^\uov G^\uov T_1F^vG^vF^uG^uA} \\
	& {F^uG^uF^uT_uG^uF^uG^uA} & {F^uG^uT_1F^uG^uA} & {F^u G^u T_1F^vG^vF^uG^uA} \\
	\\
	&& {F^uG^uT_1F^uG^uA} \\
	\\
	\\
	{F^vG^vF^uG^uT_1A} & {F^vG^vF^uG^uT_1F^uG^uA} & {F^vG^vT_1F^uG^uA} & {F^vG^vT_1F^vG^vF^uG^uA} & {T_1F^vG^vF^uG^uA}
	\arrow["{\varepsilon^\uov_{T_1 F^\uov G^\uov A}}", curve={height=-12pt}, from=1-4, to=1-5]
	\arrow[""{name=0, anchor=center, inner sep=0}, "{F^u\eta^u_{G^uT_1A}}"{description}, from=4-2, to=6-2]
	\arrow[""{name=1, anchor=center, inner sep=0}, "{F^vG^v\varepsilon^u_{T_1F^uG^uA}}"', curve={height=12pt}, from=13-2, to=13-3]
	\arrow[""{name=2, anchor=center, inner sep=0}, "{\varepsilon^v_{T_1F^vG^vF^uG^uA}}"', curve={height=12pt}, from=13-4, to=13-5]
	\arrow[""{name=3, anchor=center, inner sep=0}, "{F^uG^uT_1(\alpha_{A,U,V})}"{description}, from=5-4, to=8-4]
	\arrow[""{name=4, anchor=center, inner sep=0}, "{F^u\varepsilon^{\uov\leq u}_{G^u T_1 F^\uov G^\uov A}}"{description}, from=1-4, to=5-4]
	\arrow[""{name=5, anchor=center, inner sep=0}, "{T_1(\alpha_{A,U,V})}"{description, pos=0.6}, shift left=15, curve={height=-30pt}, from=1-5, to=13-5]
	\arrow[""{name=6, anchor=center, inner sep=0}, "{F^v\varepsilon^{\uov\leq v}_{G^v T_1F^uG^uF^uG^uA}}"{description, pos=0.4}, from=7-5, to=13-4]
	\arrow["{\varepsilon^\uov_{T_1F^uG^uF^uG^uA}}"{description, pos=0.6}, from=7-5, to=13-5]
	\arrow[""{name=7, anchor=center, inner sep=0}, "{F^u\varepsilon^{\uov \leq u}_{G^u T_1F^vG^vF^uG^uA}}"{description}, from=7-5, to=8-4]
	\arrow[""{name=8, anchor=center, inner sep=0}, "{F^v\varepsilon^{u\leq v}_{G^vT_1F^vG^vF^uG^uA}}"{description}, from=8-4, to=13-4]
	\arrow[""{name=9, anchor=center, inner sep=0}, "{F^uG^uT_1\varepsilon^v_{F^uG^uA}}"', from=8-4, to=10-3]
	\arrow["{F^v\varepsilon^{u\leq v}_{G^vT_1F^uG^uA}}"{description}, from=10-3, to=13-3]
	\arrow["{F^uG^uF^uT_u\eta^u_{G^uA}}"{pos=0.4}, curve={height=-12pt}, from=6-2, to=8-2]
	\arrow["{F^\uov G^\uov T_1F^u\varepsilon^{\uov\leq u}_{G^uA}}"{description, pos=0.6}, from=1-4, to=4-3]
	\arrow["{F^uG^uF^uT_uG^u\varepsilon^u_A}", curve={height=-12pt}, from=8-2, to=6-2]
	\arrow[""{name=10, anchor=center, inner sep=0}, Rightarrow, no head, from=8-3, to=10-3]
	\arrow[""{name=11, anchor=center, inner sep=0}, "{F^\uov G^\uov T_1\varepsilon^u_A}"'{pos=0.6}, curve={height=6pt}, from=4-3, to=1-3]
	\arrow["{F^uF^{\uov\leq u}G^{\uov\leq u}T_u\eta^u_{G^uA}}"'{pos=0.6}, curve={height=6pt}, from=1-3, to=4-3]
	\arrow["{F^vG^vF^uT_uG^u\varepsilon^u_A}"', curve={height=12pt}, from=13-2, to=13-1]
	\arrow["{F^vG^vF^uT_u\eta^u_{G^uA}}"', curve={height=12pt}, from=13-1, to=13-2]
	\arrow[""{name=12, anchor=center, inner sep=0}, "{F^v\varepsilon^{u\leq v}_{G^vF^uT_uG^uF^uG^uA}}"{description}, from=8-2, to=13-2]
	\arrow[""{name=13, anchor=center, inner sep=0}, "{F^vG^vT_1\varepsilon^v_{F^uG^uA}}"', curve={height=12pt}, from=13-4, to=13-3]
	\arrow["{F^vT_v\eta^v_{G^vF^uG^uA}}"', curve={height=12pt}, from=13-3, to=13-4]
	\arrow[""{name=14, anchor=center, inner sep=0}, "{F^u\varepsilon^{\uov\leq u}_{T_uG^uF^uG^uA}}"{description}, from=4-3, to=6-3]
	\arrow[Rightarrow, no head, from=6-3, to=8-3]
	\arrow[""{name=15, anchor=center, inner sep=0}, "{F^uT_u\eta^u_{G^uA}}"{description}, from=4-2, to=6-3]
	\arrow[""{name=16, anchor=center, inner sep=0}, "{\varepsilon^v_{F^uG^uT_1F^uG^uA}}"{description}, from=13-2, to=10-3]
	\arrow["{F^u\varepsilon^u_{T_1F^uG^uA}}"', curve={height=12pt}, from=8-2, to=8-3]
	\arrow[""{name=17, anchor=center, inner sep=0}, "{F^u\eta^u_{T_uG^uF^uG^uA}}"', curve={height=12pt}, from=8-3, to=8-2]
	\arrow[""{name=18, anchor=center, inner sep=0}, "{\alpha_{T_1A,U,V}}"{description}, from=1-1, to=6-1]
	\arrow["{F^v\varepsilon^{u\leq v}_{G^vF^uG^uT_1A}}"', curve={height=12pt}, from=6-2, to=6-1]
	\arrow[""{name=19, anchor=center, inner sep=0}, Rightarrow, no head, from=6-1, to=13-1]
	\arrow[""{name=20, anchor=center, inner sep=0}, "{F^\uov G^\uov T_1(\alpha_{A,U,V})}"{description}, from=1-4, to=7-5]
	\arrow["{F^u\varepsilon^{\uov \leq v}_{G^uT_1A}}"{description}, from=1-1, to=4-2]
	\arrow[""{name=21, anchor=center, inner sep=0}, Rightarrow, no head, from=1-1, to=1-3]
	\arrow["{F^\uov G^\uov T_1\varepsilon^{\uov}_A}", curve={height=-6pt}, from=1-4, to=1-3]
	\arrow["{F^\uov T_\uov \eta^\uov_{G^\uov A}}", curve={height=-12pt}, from=1-3, to=1-4]
	\arrow[""{name=22, anchor=center, inner sep=0}, "{F^uG^uT_1F^u\varepsilon^{\uov\leq u}_{G^uA}}"{description}, from=5-4, to=6-3]
	\arrow["{\textbf{nat. }\mathbf{\varepsilon}}"{description}, shift left=5, Rightarrow, draw=none, from=20, to=5]
	\arrow["{\textbf{decomp. }\mathbf{\varepsilon}}"{description}, shift left=5, Rightarrow, draw=none, from=6, to=2]
	\arrow["{\textbf{def. }\mathbf{\alpha}}"{description}, Rightarrow, draw=none, from=18, to=0]
	\arrow["{\textbf{nat. }\mathbf{\varepsilon}}"{description, pos=0.6}, Rightarrow, draw=none, from=21, to=15]
	\arrow["{\textbf{nat. }\mathbf{\eta}}"{description}, Rightarrow, draw=none, from=15, to=17]
	\arrow["{\textbf{nat. }\mathbf{\varepsilon}}"{description}, Rightarrow, draw=none, from=12, to=10]
	\arrow["{\textbf{nat. }\mathbf{\varepsilon}}"{description}, Rightarrow, draw=none, from=9, to=13]
	\arrow["{\textbf{decomp. }\mathbf{\varepsilon}}"{description}, shift left=5, Rightarrow, draw=none, from=16, to=1]
	\arrow["{\textbf{(a)}}"{description}, Rightarrow, draw=none, from=22, to=9]
	\arrow["{\textbf{nat. }\mathbf{\varepsilon}}"{description}, shift right=5, Rightarrow, draw=none, from=3, to=20]
	\arrow["{\textbf{nat. }\mathbf{\varepsilon}}"{description}, shift right=5, Rightarrow, draw=none, from=14, to=4]
	\arrow["{\textbf{decomp. }\mathbf{\varepsilon}}"{description}, Rightarrow, draw=none, from=11, to=1-4]
	\arrow["{\textbf{nat. }\mathbf{\varepsilon}}"{description}, shift left=5, Rightarrow, draw=none, from=19, to=12]
	\arrow["{\textbf{decomp. }\mathbf{\varepsilon}}"{description, pos=0.6}, shift right=1, Rightarrow, draw=none, from=8, to=7]
\end{tikzcd}\]
\caption{Commuting diagram establishing axiom~\eqref{eq:st:associator} in the proof of Proposition \ref{thm:formal_to_loc}.}
\label{fig:axiom2_proof}
 \end{sidewaysfigure}

\item We move to axiom~\eqref{eq:st:unit}: $(\eta\restrict{1})_{A\otimes U} = \str_{A,U}\circ ((\eta\restrict{1})_A\otimes U)$. This holds because the diagram below commutes, where we use naturality  and the zigzag equation $\varepsilon_{FGA}\circ F\eta_{GA} = \id_{FGA}$:
\[\begin{tikzcd}
	{FGT_1A} & {FT_uGA} &&& {FT_uGFGA} \\
	FGA && FGFGA && {FGT_1FGA} \\
	&&&& {T_1FGA}
	\arrow[Rightarrow, no head, from=1-5, to=2-5]
	\arrow[Rightarrow, no head, from=1-1, to=1-2]
	\arrow["{FG(\eta\restrict{1})_A}", from=2-1, to=1-1]
	\arrow[""{name=0, anchor=center, inner sep=0}, "{FT_u\eta_A}", from=1-2, to=1-5]
	\arrow[""{name=1, anchor=center, inner sep=0}, "{\varepsilon_{T_1FGA}}", from=2-5, to=3-5]
	\arrow["{F\eta_{GA}}"{description, pos=0.6}, curve={height=-6pt}, from=2-1, to=2-3]
	\arrow[""{name=2, anchor=center, inner sep=0}, "{\varepsilon_{FGA}}"{description, pos=0.6}, curve={height=-6pt}, from=2-3, to=2-1]
	\arrow["{FG (\eta\restrict{1})_{FGA}}", from=2-3, to=2-5]
	\arrow["{(\eta\restrict{1})_{FGA}}"{description}, shift right=1, curve={height=18pt}, from=2-1, to=3-5]
	\arrow["{\textbf{nat. } \mathbf{(\eta\restrict{1})}}"{description}, shift left=3, curve={height=6pt}, Rightarrow, draw=none, from=0, to=2-3]
	\arrow["{\textbf{nat. } \mathbf{\varepsilon}}"{description}, shift left=4, Rightarrow, draw=none, from=1, to=2]
\end{tikzcd}\]
Here $(\eta\restrict{1})$ refers to the unit of the monad $T_1$, while $\eta$ (and $\varepsilon$) refers to the unit (and counit) of the adjunction $F\dashv G$.

\item Equation~\eqref{eq:st:multiplication}, that $(\mu\restrict{1})_{A \otimes U} \circ T_1(\str_{A,U}) \circ \str_{T_1(A),U} = \str_{A,U} \circ ((\mu\restrict{1})_A \otimes U)$, follows from  commutivity of the diagram below:
\[\begin{tikzcd}[column sep = 3 mm, row sep = 8mm]
	{FGT_1T_1A} &&&& {FGT_1A} \\
	{FT_uGT_1A} && {FT_uT_uGA} && {FT_uGA} \\
	{FT_uGFGT_1A} & {FT_uGFT_uGA} \\
	{FGT_1FGT_1A} & {FGT_1FT_uGA} && {FT_uT_uGFGA} & {FT_uGFGA} \\
	&& {FT_uGFT_uGFGA} \\
	{T_1FGT_1A} && {FT_uGFGT_1FGA} & {FT_uGT_1FGA} & {FGT_1FGA} \\
	& {FGT_1FT_uGFGA} & {FGT_1FGT_1FGA} & {FGT_1T_1FGA} \\
	\\
	{T_1FT_uGA} & {T_1FT_uGFGA} & {T_1FGT_1FGA} & {T_1T_1FGA} & {T_1FGA}
	\arrow[""{name=0, anchor=center, inner sep=0}, "{FG(\mu\restrict{1})_A}", from=1-1, to=1-5]
	\arrow[Rightarrow, no head, from=1-1, to=2-1]
	\arrow["{FT_u\eta_{GT_1A}}"{description}, from=2-1, to=3-1]
	\arrow[Rightarrow, no head, from=3-1, to=4-1]
	\arrow["{\varepsilon_{T_1FGT_1A}}"{description}, from=4-1, to=6-1]
	\arrow[Rightarrow, no head, from=6-1, to=9-1]
	\arrow[""{name=1, anchor=center, inner sep=0}, "{T_1GT_u\eta_{GA}}"', from=9-1, to=9-2]
	\arrow[Rightarrow, no head, from=9-2, to=9-3]
	\arrow[""{name=2, anchor=center, inner sep=0}, "{T_1\varepsilon_{T_1FGA}}"', from=9-3, to=9-4]
	\arrow[""{name=3, anchor=center, inner sep=0}, "{(\mu\restrict{1})_{FGA}}"', from=9-4, to=9-5]
	\arrow["{\varepsilon_{T1FGA}}"{description}, from=6-5, to=9-5]
	\arrow[""{name=4, anchor=center, inner sep=0}, Rightarrow, no head, from=4-5, to=6-5]
	\arrow["{FT_u\eta_{GA}}"{description}, from=2-5, to=4-5]
	\arrow[Rightarrow, no head, from=1-5, to=2-5]
	\arrow[Rightarrow, no head, from=2-1, to=2-3]
	\arrow[""{name=5, anchor=center, inner sep=0}, "{F(\mu\restrict{u})_{GA}}"{description}, from=2-3, to=2-5]
	\arrow[""{name=6, anchor=center, inner sep=0}, "{F(\mu\restrict{u})_{GFGA}}", curve={height=-12pt}, from=4-4, to=4-5]
	\arrow["{FT_uT_u\eta_{GA}}"{description}, from=2-3, to=4-4]
	\arrow[""{name=7, anchor=center, inner sep=0}, Rightarrow, no head, from=4-4, to=6-4]
	\arrow[Rightarrow, no head, from=6-4, to=7-4]
	\arrow[""{name=8, anchor=center, inner sep=0}, "{GF(\mu\restrict{1})_{FGA}}"{description, pos=0.6}, curve={height=6pt}, from=7-4, to=6-5]
	\arrow["{\varepsilon_{T_1T_1FGA}}"{description}, from=7-4, to=9-4]
	\arrow[Rightarrow, no head, from=3-1, to=3-2]
	\arrow[Rightarrow, no head, from=3-2, to=4-2]
	\arrow[Rightarrow, no head, from=4-1, to=4-2]
	\arrow["{\varepsilon_{T_1FT_uGA}}"{description}, from=4-2, to=9-1]
	\arrow[""{name=9, anchor=center, inner sep=0}, "{FGT_1\varepsilon_{T_1FGA}}"', curve={height=12pt}, from=7-3, to=7-4]
	\arrow["{\varepsilon_{T_1FGT_1FGA}}"{description}, from=7-3, to=9-3]
	\arrow[Rightarrow, no head, from=7-2, to=7-3]
	\arrow[""{name=10, anchor=center, inner sep=0}, "{FGT_1FT_u\eta_{GA}}"{description}, from=4-2, to=7-2]
	\arrow["{\varepsilon_{T_1FT_uGFGA}}"{description}, from=7-2, to=9-2]
	\arrow[""{name=11, anchor=center, inner sep=0}, "{FT_u\eta_{T_uGA}}"{description}, from=2-3, to=3-2]
	\arrow[Rightarrow, no head, from=6-3, to=7-3]
	\arrow[Rightarrow, no head, from=5-3, to=6-3]
	\arrow["{FT_uFT_u\eta_{GA}}"{description, pos=0.7}, from=3-2, to=5-3]
	\arrow[Rightarrow, no head, from=5-3, to=7-2]
	\arrow["{FT_u\eta_{GT_1FGA}}", curve={height=-12pt}, tail reversed, no head, from=6-3, to=6-4]
	\arrow[""{name=12, anchor=center, inner sep=0}, "{FT_u\eta_{T_uGFGA}}"{description}, from=4-4, to=5-3]
	\arrow["{FT_uG\varepsilon_{T_1FGA}}"', curve={height=12pt}, from=6-3, to=6-4]
	\arrow["{\textbf{nat. }\mathbf{\varepsilon}}"{description}, Rightarrow, draw=none, from=9, to=2]
	\arrow["{\textbf{nat. }\mathbf{\varepsilon}}"{description}, Rightarrow, draw=none, from=8, to=3]
	\arrow["{\textbf{nat. }\mathbf{\varepsilon}}"{description, pos=0.6}, Rightarrow, draw=none, from=10, to=1]
	\arrow["{\textbf{nat. }\mathbf{\eta}}"{description}, Rightarrow, draw=none, from=11, to=12]
	\arrow["{\textbf{(a)}}"{description}, shift right=1, Rightarrow, draw=none, from=0, to=2-3]
	\arrow["{\textbf{nat. }\mathbf{\mu}}"{description}, Rightarrow, draw=none, from=5, to=6]
	\arrow["{\textbf{(a)}}"{description}, Rightarrow, draw=none, from=7, to=4]
\end{tikzcd}\]
Here, most polygons are strictly equal due to our formal monad being a natural transformation (equations~\eqref{eq:equal_ob} and \eqref{eq:nat_on_mor}). Some follow from naturality, while those marked $(a)$ follow from $\mu$ being a modification (equation~\eqref{eq:modification}).

\item Axiom \eqref{eq:st:naturalinu}, which reads $\str_{A,V} \circ (T(A) \otimes m)  = T(A \otimes m) \circ \str_{A,U}$, is established as follows:
\[\begin{tikzcd}[column sep = 4mm, row sep = 6mm]
	{F^uG^uT_1A} & {F^uT_uG^uA} &&& {F^uT_uG^uF^uG^uA} & {F^uG^uT_1F^uG^uA} &&& {T_1F^uG^uA} \\
	\\
	&&&& {F^vT_vG^vF^uG^uA} & {F^vG^vT_1F^uG^uA} \\
	\\
	{F^vG^vT_1A} & {F^vT_vG^vA} &&& {F^vT_vG^vF^vG^vA} & {F^vG^vT_1F^vG^vA} &&& {T_1F^vG^vA}
	\arrow[""{name=0, anchor=center, inner sep=0}, Rightarrow, no head, from=1-1, to=1-2]
	\arrow[""{name=1, anchor=center, inner sep=0}, Rightarrow, no head, from=1-5, to=1-6]
	\arrow[""{name=2, anchor=center, inner sep=0}, Rightarrow, no head, from=5-5, to=5-6]
	\arrow["{F^v\varepsilon^{\utv}_{G^vT_1A}}"{description}, from=1-1, to=5-1]
	\arrow[""{name=3, anchor=center, inner sep=0}, "{\varepsilon^v_{T_1F^vG^vA}}"', from=5-6, to=5-9]
	\arrow[""{name=4, anchor=center, inner sep=0}, Rightarrow, no head, from=5-1, to=5-2]
	\arrow[""{name=5, anchor=center, inner sep=0}, "{F^vT_v\eta^v_{G^vA}}"', from=5-2, to=5-5]
	\arrow["{T_1F^v\varepsilon^{\utv}_{G^vA}}"{description}, from=1-9, to=5-9]
	\arrow[""{name=6, anchor=center, inner sep=0}, "{\varepsilon^u_{T_1F^uG^uA}}", from=1-6, to=1-9]
	\arrow[""{name=7, anchor=center, inner sep=0}, "{F^uT_u\eta^u_{G^uA}}", from=1-2, to=1-5]
	\arrow["{F^v\varepsilon^{\utv}_{G^vT_1F^uG^uA}}"{description}, from=1-6, to=3-6]
	\arrow["{F^vG^vT_1F^v\varepsilon^{\utv}_{G^vA}}"{description}, from=3-6, to=5-6]
	\arrow["{F^vT_vG^vF^v\varepsilon^{\utv}_{G^vA}}"{description}, from=3-5, to=5-5]
	\arrow["{F^v\varepsilon^{\utv}_{T_vG^vA}}"{description}, from=1-2, to=5-2]
	\arrow["{F^v\varepsilon^{\utv}_{T_vG^vF^uG^uA}}"{description}, from=1-5, to=3-5]
	\arrow["{\textbf{(b)}}"{description}, Rightarrow, draw=none, from=7, to=5]
	\arrow["{\textbf{(a)}}"{description}, Rightarrow, draw=none, from=0, to=4]
	\arrow["{\textbf{(c)}}"{description}, Rightarrow, draw=none, from=1, to=2]
	\arrow["{\textbf{(d)}}"{description}, Rightarrow, draw=none, from=6, to=3]
\end{tikzcd}\]
The left square \textbf{(a)} is technically an equality from~\eqref{eq:equal_ob}, while \textbf{(c)} also trivially follows from~\eqref{eq:equal_ob}. Squares \textbf{(b)} and \textbf{(d)} follow from the naturality of the counit and the decomposition property of~\eqref{eq:decomp_counit}.

\item Finally, equation~\eqref{eq:st:naturalina} asks that strength is natural: $\str_{B,U} \circ \big( T(f) \otimes U \big) = T(f \otimes U) \circ \str_{A,U}$. 
\[\begin{tikzcd}[column sep = 6mm, row sep = 12mm]
	{FGT_1A} & {FT_uGA} && {FT_uGFGA} & {FGT_1FGA} && {T_1FGA} \\
	{FGT_1B} & {FT_uGB} && {FT_uGFGB} & {FGT_1FGB} && {T_1FGB}
	\arrow[Rightarrow, no head, from=1-1, to=1-2]
	\arrow[""{name=0, anchor=center, inner sep=0}, "{FT_u\eta_{GA}}", from=1-2, to=1-4]
	\arrow[Rightarrow, no head, from=1-4, to=1-5]
	\arrow[""{name=1, anchor=center, inner sep=0}, "{\varepsilon_{T_1FGA}}", from=1-5, to=1-7]
	\arrow["{T_1FGf}", from=1-7, to=2-7]
	\arrow[""{name=2, anchor=center, inner sep=0}, "{\varepsilon_{T_1FGB}}"', from=2-5, to=2-7]
	\arrow[Rightarrow, no head, from=2-4, to=2-5]
	\arrow[""{name=3, anchor=center, inner sep=0}, "{FT_u\eta_{GB}}"', from=2-2, to=2-4]
	\arrow[Rightarrow, no head, from=2-1, to=2-2]
	\arrow["{FGT_1f}"', from=1-1, to=2-1]
	\arrow["{FT_uGf}"{description}, from=1-2, to=2-2]
	\arrow["{FT_uGFGf}"{description}, from=1-4, to=2-4]
	\arrow["{FGT_1FGf}"{description}, from=1-5, to=2-5]
	\arrow["{\textbf{\text{nat. }} \mathbf{\eta}}"{description}, Rightarrow, draw=none, from=0, to=3]
	\arrow["{\textbf{nat. } \mathbf{\varepsilon}}"{description}, Rightarrow, draw=none, from=1, to=2]
\end{tikzcd}\]
\end{enumerate}
This follows from naturality of the unit and counit. 
\end{proof}

\restate{prop:loc_to_formal}
\begin{proposition}
  A localisable monad $T$ on a stiff category $\cat{C}$ induces a formal monad on $\overline{\cat{C}}$ in $[\ZI(\cat{C})\op,\cat{Cat}]$.
  The natural transformation $\overline{T} \colon \overline{\cat{C}} \Rightarrow \overline{\cat{C}}$ has components $T\restrict{u}$, the modification $\overline{\eta} \colon \overline{\cat{C}} \Rrightarrow \overline{T}$ has components $\eta\restrict{u}$, and the modification $\overline{\mu} \colon \overline{T}^2 \Rrightarrow \overline{T}$ has components $\mu\restrict{u}$ as in Proposition~\ref{prop:smallmonad}.
\end{proposition}
\begin{proof}
  To see that $\overline{T}$ is natural, suppose $u = v \circ m$. We are to show that the following diagram in $\cat{Cat}$ commutes:
\[\begin{tikzcd}[column sep =4 mm, row sep = 2mm]
	{\cC\restrict{v}} && {\cC\restrict{v}} \\
	\\
	{\cC\restrict{u}} && {\cC\restrict{u}}
	\arrow["{T\restrict{v}}", from=1-1, to=1-3]
	\arrow["{T\restrict{u}}"', from=3-1, to=3-3]
	\arrow["\cC\restrict{u\leq v}"', from=1-1, to=3-1]
	\arrow["\cC\restrict{u\leq v}", from=1-3, to=3-3]
\end{tikzcd}\]
  On objects this is clear, because the vertical functors act as the identity on objects, and the horizontal functors act as $T$ on objects. Let $f \colon A \otimes V \to B$ be a morphism $A \to B$ in $\cat{C}\restrict{v}$. Mapping it along the left-bottom path sends it first to $f \circ (A \otimes m) \colon A \otimes U \to B$ and finally to $T(f) \circ T(A \otimes m) \circ \str_{A,U}$ in $\cat{C}\restrict{u}$.
  Mapping it along the top-right path sends it to $T(f) \circ \str_{A,V} \circ (T(A) \otimes m)$. But these two morphisms are equal by~\eqref{eq:st:naturalinu}.

  That $\overline{\eta}$ and $\overline{\mu}$ are modifications comes down to the components $\eta\restrict{u}$ and $\mu\restrict{u}$ satisfying the equations of~\eqref{eq:modification}. But this follows directly from Lemma \ref{lem:monadmorphism}. Naturality of $(\eta\restrict{u})_A \colon A \to T\restrict{u}(A)$ in $\cat{C}\restrict{u}$ means that if $f \colon A \otimes U \to B$ in $\cat{C}$, then we must have $(\eta\restrict{u})_B \circ (f \otimes U) \circ (A \otimes U \otimes u)^{-1} = T(f) \circ \str_{T(A),U} \circ ((\eta\restrict{u})_A \otimes U) \circ (A \otimes U \otimes u)^{-1}$. This is indeed the case by~\eqref{eq:st:unit} as shown in the diagram on the left below. It remains to verify that ${(\mu\restrict{u})_A}$ is natural in $\cat{C}\restrict{u}$. This means that for $f \colon A \otimes U \to B$ in $\cat{C}$ the following diagram on the right must commute in $\cat{C}$.  But this follows from~\eqref{eq:st:multiplication}, bifunctoriality of the tensor, and naturality of $\mu$.
\[\begin{tikzcd}[column sep = 4 mm, row sep = 4 mm]
	{A \otimes U \otimes U} &&&& {T(A) \otimes U} \\
	\\
	& {A \otimes U} &&& {T(A \otimes U)} \\
	& B \\
	{B \otimes U} &&&& {T(B)}
	\arrow["{f \otimes U}"', from=1-1, to=5-1]
	\arrow["{A \otimes U \otimes u}"{description}, from=1-1, to=3-2]
	\arrow["{B \otimes u}"{description}, from=5-1, to=4-2]
	\arrow["f"', from=3-2, to=4-2]
	\arrow["{\eta_B \otimes u}"', from=5-1, to=5-5]
	\arrow["{\eta_A \otimes u \otimes U}", from=1-1, to=1-5]
	\arrow["{\str_{A,U}}", from=1-5, to=3-5]
	\arrow["{\eta_{A \otimes U}}"{description}, from=3-2, to=3-5]
	\arrow["{\eta_A \otimes U}"{description}, from=3-2, to=1-5]
	\arrow["{\eta_B}"{description}, from=4-2, to=5-5]
	\arrow["{T(f)}", from=3-5, to=5-5]
\end{tikzcd}
\qquad
\begin{tikzcd}[column sep =4 mm]
	{T^2(A) \otimes U \otimes U} && {T(A) \otimes U} \\
	{T(T(A) \otimes U) \otimes U} \\
	{T^2(A \otimes U) \otimes U} && {T(A \otimes U)} \\
	{T^2(B) \otimes U} && {T(B)}
	\arrow["{\str_{T(A),U} \otimes U}"', from=1-1, to=2-1]
	\arrow["{T(\str_{A,U})\otimes U}"', from=2-1, to=3-1]
	\arrow["{T(f) \otimes U}"', from=3-1, to=4-1]
	\arrow["{\mu_B \otimes u}"', from=4-1, to=4-3]
	\arrow["{\mu_{A} \otimes u \otimes U}", from=1-1, to=1-3]
	\arrow["{\str_{A,U}}", from=1-3, to=3-3]
	\arrow["{T(f)}", from=3-3, to=4-3]
	\arrow["{\mu_{A \otimes U} \otimes u}", from=3-1, to=3-3]
\end{tikzcd}\]
  Finally, that $\overline{\eta}$ and $\overline{\mu}$ satisfy the monad laws (pointwise) follows from Proposition~\ref{prop:smallmonad}.
\end{proof}

\end{document}